\documentclass[10pt, a4paper]{article}

\usepackage{styles/arxiv}         
\usepackage[utf8]{inputenc}

\usepackage[english]{babel}
\usepackage[T1]{fontenc}    
\usepackage{hyperref}       
\usepackage{url}            
\usepackage{booktabs}       
\usepackage{amsfonts}       
\usepackage{microtype}      
\usepackage{amsmath}
\usepackage{amssymb}        
\usepackage{amsthm}         
\usepackage{dsfont}         
\usepackage{float}
\usepackage{graphicx}
\usepackage{stmaryrd}       
\usepackage{tikz}
\usepackage{bbm}

\usepackage{pgfplots}
\usepgfplotslibrary{colormaps}

\usepackage{csquotes}
\usepackage{todonotes}
\usepackage{soul}           
\usepackage{marvosym}       
\usepackage{caption}
\usepackage{subcaption}

\usepackage{algorithm}
\usepackage{algpseudocode}

\usepackage{lipsum}

\def\Var{{\rm Var}}

\usepackage{multirow}
\newcommand{\indep}{\perp \!\!\! \perp}
\DeclareMathOperator*{\argmin}{arg\,min}
\allowdisplaybreaks
\usepackage{natbib}
\setcitestyle{authoryear}
\newtheorem{theorem}{Theorem}
\newtheorem{lemma}{Lemma}
\usepackage{textcomp,eurosym}

\usetikzlibrary{trees}
\usetikzlibrary{arrows,automata}

\usepackage[edges]{forest}

\forestset{
  direction switch/.style={
    for tree={edge+=thick, font=\sffamily},
    where level>=1{folder, grow'=0}{for children=forked edge},
    where level=3{}{draw},
  },
}

\let\oldReturn\Return
\renewcommand{\Return}{\State\oldReturn}

\pgfplotsset{compat=1.18}
\pgfplotsset{
    colormap={uni}{rgb255(0cm)=(0, 170, 220); rgb255(1cm)=(255, 20, 11)}
}

\usepgfplotslibrary{fillbetween}

\definecolor{myblue}{RGB}{0, 170, 220}

\graphicspath{ {./figures/} }

\title{Assumption-Lean Quantile Regression}

\date{April, 2024} 					

\title{Assumption-Lean Quantile Regression}
\author{Georgi Baklicharov\\
Department of Applied Mathematics, Computer Science and Statistics\\
Ghent University\\
Krijgslaan 281 (S9), 9000 Ghent, Belgium\\
\href{mailto:georgi.baklicharov@ugent.be}{
\texttt{georgi.baklicharov@ugent.be}}
\And
Christophe Ley\\
Department of Mathematics\\
Université du Luxembourg\\
Maison du Nombre, Avenue de la Fonte 6, L-4364 Esch-sur-Alzette, Luxembourg\\
\href{mailto:georgi.baklicharov@ugent.be}{\texttt{christophe.ley@uni.lu}}
\And
Vanessa Gorasso\\
Department of Epidemiology and Public Health\\
Sciensano\\
Rue J Wytsman 14, 1050 Brussels, Belgium\\
\href{mailto:georgi.baklicharov@ugent.be}{\texttt{vanessa.gorasso@sciensano.be}}
\And
Brecht Devleesschauwer\\
Department of Epidemiology and Public Health\\
Sciensano\\
Rue J Wytsman 14, 1050 Brussels, Belgium\\
\href{mailto:georgi.baklicharov@ugent.be}{\texttt{brecht.devleesschauwer@sciensano.be}}
\And
Stijn Vansteelandt\\
Department of Applied Mathematics, Computer Science and Statistics\\
Ghent University\\
Krijgslaan 281 (S9), 9000 Ghent, Belgium\\
\href{mailto:georgi.baklicharov@ugent.be}{\texttt{stijn.vansteelandt@ugent.be}}
}

\begin{document}

\tikzstyle{every node}=[draw=black,thick,anchor=west]
\tikzstyle{selected}=[draw=red,fill=red!30]
\tikzstyle{optional}=[dashed,fill=gray!50]

\maketitle

\begin{abstract}
Quantile regression is a powerful tool for detecting exposure-outcome associations given covariates across different parts of the outcome's distribution, but has two major limitations when the aim is to infer the effect of an exposure. Firstly, the exposure coefficient estimator may not converge to a meaningful quantity when the model is misspecified, and secondly, variable selection methods may induce bias and excess uncertainty, rendering inferences biased and overly optimistic. In this paper, we address these issues via partially linear quantile regression models which parametrize the conditional association of interest, but do not restrict the association with other covariates in the model. We propose consistent estimators for the unknown model parameter by mapping it onto a nonparametric main effect estimand that captures the (conditional) association of interest even when the quantile model is misspecified. This estimand is estimated using the efficient influence function under the nonparametric model, allowing for the incorporation of data-adaptive procedures such as variable selection and machine learning. Our approach provides a flexible and reliable method for detecting associations that is robust to model misspecification and excess uncertainty induced by variable selection methods. The proposal is illustrated using simulation studies and data on annual health care costs associated with excess body weight.
\end{abstract}

\section{Introduction}
Quantile regression, originally introduced by \cite{koenkerbassett}, is a statistical method that provides a flexible framework for modeling the relationship between a response variable and predictors. Unlike traditional linear regression models that focus on estimating the conditional mean of the response variable, quantile regression allows for the estimation of conditional quantiles, which can provide valuable and broader insights into the distributional characteristics of the response variable. It has gained popularity in recent years thanks to its ability to capture heterogeneity in the data and provide more robust estimates in the presence of outliers, skewed and/or heavy-tailed distributions \citep{koenker2005quantile,koenker2017,koenker2017handbook}. Quantile regression has been applied in a variety of fields, including economics, finance, health, and environmental sciences, among others \citep[e.g.][]{koenkerhallock2001,YuLuStander} but deserves to be used more frequently, because of the rich description it provides of various relevant aspects of data distributions, combined with appealing ease of interpretation \citep{das2019quantile}.

In this work, we focus on the estimation of causal quantile treatment or exposure effects under standard causal assumptions \citep{hernan2020causal}. The quantile treatment effect has originally been defined by \cite{doksum1974} and \cite{lehmann1975} in a two-sample setting as the horizontal distance (i.e.\ the distance along the horizontal axis between two fixed points of different distribution functions) between the outcome distribution of exposed versus unexposed. However, extension to allow for covariate adjustment via standard quantile regression raises concerns due to potential model misspecification and post-selection inference. 

Model misspecification is a concern for most statistical methods. When the goal is solely prediction, this is often less problematic. Wrong models can still deliver good predictions, because predictions are commonly obtained by minimizing a measure of prediction error over the considered model class. Moreover, unbiased estimates of prediction error are obtained, even when the model is misspecified. However, when estimating exposure effects, model misspecification typically causes more severe problems. For instance, in (generalized) linear regression models that have a main exposure effect and covariates, the ordinary least squares estimator generally fails to converge to a comprehensive measure of the conditional relationship between the exposure and outcome in the presence of model misspecification \citep{AssumptionLean}. In quantile regression models, the situation is even worse as there is no good understanding of where the effect of a randomized treatment converges to when the model is wrong. One possible solution is to make the quantile model sufficiently complex by including interactions, higher-order terms, transformations, etc. However, it is generally unclear how much additional complexity is needed, and moreover, this added complexity may lead to overfitting and thus estimates with large bias and standard errors. In contrast, simple models help to make results insightful and easy to communicate. In practice, one is therefore often torn between keeping the model simple so that it is comprehensible and making it complex so that it represents the ground truth \citep{breiman2001}. A common way of finding a trade-off between model simplicity and model correctness is by using data-adaptive procedures, such as variable selection algorithms. However, this also raises concerns because data-adaptive procedures usually give rise to additional uncertainty that is not taken into account in standard error calculations. In particular, standard inference commonly ignores the fact that making decisions on whether or not to add certain covariates was data-driven, while it is highly unlikely to make the exact same decisions under repeated sampling \citep{leebpotscher2005}. This problem of post-selection inference renders standard inferences overly optimistic.

In previous literature, attempts have been made to circumvent this issue by employing nonparametric approaches, focusing on marginal effects. However, the majority of these approaches is limited to binary exposures \citep[e.g.][]{melly2006,firpo2007,chernozhukov2013,DONALD2014383,caracciolo2017,athey2021semiparametric}. \cite{Ai2022} extended the work of \cite{DONALD2014383} to continuous exposures, but still focussing on marginal quantile effects. Instead, our objective is to infer conditional quantile effects as these are less vulnerable to extrapolation (see the discussion section). Our proposal will be in line with recent developments in causal machine learning \citep{DIAZ201739,belloni2017,kallus2019localized} in the sense that we allow the use of machine learning techniques, while ensuring valid post-selection inference.

Throughout this work, we will focus on the conditional association between an outcome $Y$ and a possibly continuous exposure $A$, given a vector of measured covariates denoted by $L$ that are sufficient to control for confounding. Therefore, we consider the partially linear quantile model
\begin{align}\label{model}
    Q_\tau(Y\mid A,L) = \beta_\tau A + \omega_\tau(L),
\end{align}
where $Q_\tau(Y\mid A,L)$ denotes the $\tau$-quantile of the conditional distribution of $Y$ given $A$ and $L$ for a given value of $0<\tau<1$, $\beta_\tau$ is an unknown scalar parameter and $\omega_\tau(.)$ is an unknown function. Both $\beta_\tau$ and $\omega_\tau(.)$ may be quantile-specific. The parameter $\beta_\tau$ in this model reflects the conditional association of interest, for which we will develop a debiased machine learning estimator \citep{chernozhukov2018}. First, we will assume model (\ref{model}) to hold. Later, we explain how our estimator still captures the association of interest if the model is misspecified.

Our work builds on the large literature on partially linear quantile models \citep[e.g.][]{lee_2003,sun2005,WU20101607,WU2014170,lv2015quantile,SherwoodWang2016,nnplqr}. Traditionally, two popular approaches are being considered for estimating the nonparametric component of the partially linear quantile model: kernel estimators \citep[e.g.][]{lee_2003,sun2005,lv2015quantile}, and splines \citep[e.g.][]{he1996,SherwoodWang2016}. Those two approaches are also common in (partially linear) single-index quantile models \citep[e.g.][]{WU20101607,WU2014170}. However, those prior developments have a number of important shortcomings. First, in high-dimensional applications, they force high computational demands and challenges stemming from reliance on kernel weighting or splines. To accommodate this concern, \cite{nnplqr} proposed neural networks to estimate the nonparametric component. However, a further concern with this and the aforementioned proposals is that the probability limit of the estimators has no straightforward interpretation when the partially linear quantile model is misspecified (e.g., due to exposure-covariate interactions). In that case, existing proposals also fail to provide valid data-adaptive inference (not even for the probability limit of the estimator because it is a function of high-dimensional nuisance parameters with complex distributions). Our proposal provides a solution for both shortcomings by allowing the use of flexible estimators that are known to work well in high dimensions, and by ensuring an inference that is correct even under model misspecification. We take inspiration from related work for generalized linear models in a discussion paper by \cite{AssumptionLean} and for Cox models \citep{alcox}, which are both closely related to recent developments on Targeted Maximum Likelihood Estimation (TMLE) \citep{vanderlaanrose2011} and Debiased Machine Learning (DML) \citep{chernozhukov2018}.

\section{Data-adaptive inference for $\beta_\tau$}\label{sec:inference}
\subsection{A plug-in estimator}
To construct a debiased machine learning estimator for $\beta_\tau$ in model (\ref{model}), we use the fact that it can be expressed as
\begin{equation}\label{estimand2}
\Psi_\tau = \frac{E[(A-A^*)\{Q_\tau(Y\mid A,L)-Q_\tau(Y\mid A^*,L)\}]}{E\{(A-A^*)^2\}},
\end{equation}
when model (\ref{model}) holds (see Appendix \ref{appA} for a proof). In this model-free estimand, $A$ and $A^*$ are two exposure values independently drawn from stratum $L$. In Section \ref{sec:estimand} we justify the choice of this estimand by arguing that it is an interpretable measure of exposure effect, even if model (\ref{model}) is misspecified. Intuitively, this can already be seen because it is a weighted average of differences in quantiles of the outcome for individuals with different exposure values, but the same value for $L$; it is thus guaranteed to summarize the conditional association between $A$ and $Y$ given $L$, even if the model is wrong. For this nonparametric estimand, we will develop nonparametric inference. 

The model-free estimand $\Psi_\tau$ can also be written as (see Appendix \ref{appA})
\begin{equation}\label{estimand3}
    \frac{E\left(\{A-E(A\mid L)\}\left[Q_\tau(Y\mid A,L)-E\{Q_\tau(Y\mid A,L)\mid L\}\right]\right)}{E[\{A-E(A\mid L)\}^2]}.
\end{equation}
A so-called plug-in estimator is now straightforwardly obtained by substituting $E(A\mid L)$, $Q_\tau(Y\mid A,L)$ and $E\{Q_\tau(Y\mid A,L)\mid L\}$ by consistent estimators $\hat{E}(A\mid L)$, $\hat{Q}_\tau(Y\mid A,L)$ and $\hat{E}\{\hat{Q}_\tau(Y\mid A,L)\mid L\}$, respectively, and by replacing population expectations by sample averages. This results in
\begin{equation}\label{plugin}
    \frac{\sum_{i=1}^n\left(\{A_i-\hat{E}(A_i\mid L_i)\}\left[\hat{Q}_\tau(Y_i\mid A_i,L_i)-\hat{E}\{\hat{Q}_\tau(Y_i\mid A_i,L_i)\mid L_i\}\right]\right)}{\sum_{i=1}^n\{A_i-\hat{E}(A_i\mid L_i)\}^2}.
\end{equation}
To estimate the nuisance parameters, we will use flexible data-adaptive procedures in line with recent developments on debiased machine learning (DML) \citep{vanderlaanrose2011,chernozhukov2018}. These procedures may include variable selection techniques for quantile regression, as well as more advanced machine learning algorithms such as random forests or neural networks.

However, plug-in estimators like (\ref{plugin}) typically give rise to two problems. On the one hand, data-adaptive procedures are optimized towards minimal prediction error, but prediction is not our aim. Instead, our aim is to obtain well-behaved estimators for this estimand. For example, a poor selection of variables may induce confounding or selection bias. On the other hand, data-adaptive estimators $\hat{E}(A\mid L)$, $\hat{Q}_\tau(Y\mid A,L)$ and $\hat{E}\{\hat{Q}_\tau(Y\mid A,L)\mid L\}$ have a nonstandard distribution, which is usually also unknown. This makes quantifying the uncertainty in these estimators extremely difficult. Our proposal addresses these concerns.

\subsection{Debiased Machine Learning (DML)}
In view of the aforementioned concerns, we aim to establish inference for our main effect estimand (\ref{estimand2}) under the nonparametric model. This enables us to avoid dependence on modeling constraints, which could result in excessively optimistic inferences as well as computational difficulties. Like other debiased machine learning strategies, we will make use of the efficient influence function of the estimand to accomplish this goal \citep{Pfanzagl1990,bickel1993efficient}. This is a mean zero functional of the observed data and the data-generating distribution, which characterizes the estimand's sensitivity to arbitrary (smooth) changes in the data-generating law. Theorem \ref{theorem1}, of which the proof can be found in Appendix \ref{appB}, gives the efficient influence function of the estimand (\ref{estimand3}).
\begin{theorem}\label{theorem1}
The efficient influence function of the main effect estimand $\Psi_\tau$, defined by (\ref{estimand3}), under the nonparametric model is
\begin{align}\label{eif}
\begin{split}
    \phi(Y,A,L)&=\frac{A-E(A\mid L)}{E[\{A-E(A\mid L)\}^2]}\bigg[Q_\tau(Y\mid A,L) - E\left\{Q_\tau(Y\mid A,L)\mid L\right\}\\
    &\qquad\qquad+ \frac{\tau-I\{Y\leq Q_\tau(Y\mid A,L)\}}{f\{Q_\tau(Y\mid A,L)\mid A,L\}} - \Psi_\tau \{A-E(A\mid L)\} \bigg],
\end{split}
\end{align}
where $I(.)$ and $f(.\mid A,L)$ denote an indicator function and the conditional density of $Y$ given $A$ and $L$, respectively.
\end{theorem}
If the infinite-dimensional nuisance parameters $E(A\mid L)$, $Q_\tau(Y\mid A,L)$, $E\left\{Q_\tau(Y\mid A,L)\mid L\right\}$ and $f\{Q_\tau(Y\mid A,L)\mid A,L\}$ indexing the efficient influence function were known, then it would follow from its mean zero property that a consistent estimator $\Tilde{\Psi}_\tau$ of $\Psi_\tau$ could be obtained as the value of $\Psi_\tau$ that makes the sample average of the influence functions zero. The resulting estimator’s asymptotic distribution would be governed by this influence function in the sense that
\begin{align}\label{asymplin}
    \sqrt{n}(\Tilde{\Psi}_\tau-\Psi_\tau) &= \frac{1}{\sqrt{n}}\sum_{i=1}^n \frac{A_i-E(A_i\mid L_i)}{E[\{A_i-E(A_i\mid L_i)\}^2]}\bigg[Q_\tau(Y_i\mid A_i,L_i) - E\left\{Q_\tau(Y_i\mid A_i,L_i)\mid L_i\right\}\\
    &\qquad+ \frac{\tau-I\{Y_i\leq Q_\tau(Y_i\mid A_i,L_i)\}}{f\{Q_\tau(Y_i\mid A_i,L_i)\mid A_i,L_i\}} - \Psi_\tau \{A_i-E(A_i\mid L_i)\} \bigg] + o_p(1).
\end{align}
It follows that the estimator $\Tilde{\Psi}_\tau$ is asymptotically normal by the central limit theorem. Moreover, its bias approaches zero at a faster rate than the standard error, and its variance can be estimated by computing the sample variance of the efficient influence function. However, the estimator $\Tilde{\Psi}_\tau$ is practically infeasible because the efficient influence function involves unknown nuisance parameters. We will therefore substitute these with consistent data-adaptive estimators and denote the resulting estimator $\hat{\Psi}_\tau$. In this way, we obtain the closed-form estimator
\begin{align}\label{DML}
\begin{split}
    \hat{\Psi}_\tau &= \frac{1}{n}\sum_{i=1}^n\frac{A_i-\hat{E}(A_i\mid L_i)}{\frac{1}{n}\sum_{i=1}^n\{A_i-\hat{E}(A_i\mid L_i)\}^2}\Bigg[\hat{Q}_\tau(Y_i\mid A_i,L_i)\\
    &\qquad\qquad\qquad- \hat{E}\{\hat{Q}_\tau(Y_i\mid A_i,L_i)\mid L_i\} + \frac{\tau-I\{Y_i\leq \hat{Q}_\tau(Y_i\mid A_i,L_i)\}}{\hat{f}\{\hat{Q}_\tau(Y_i\mid A_i,L_i)\mid A_i,L_i\}} \Bigg].
\end{split}
\end{align}

From the general theory on nonparametric estimation \citep{Pfanzagl1990,bickel1993efficient} and the results in the Appendix \ref{appB}, it follows that, when the nuisance parameters are estimated on a sample independent to the one on which $\hat{\Psi}_\tau$ is evaluated (to mitigate overfitting bias) and under specific convergence rate assumptions detailed in Theorem 2, the estimator is asymptotically linear in the sense that (\ref{asymplin}) holds with $\hat{\Psi}_\tau$ instead of $\Tilde{\Psi}_\tau$. Consequently, its variance can be consistently estimated as 1 over $n$ times the sample variance of the influence functions, i.e.\ $1/n^2\sum_{i=1}^n\hat{\phi}(O_i)^2$,
where $\hat{\phi}$ is defined like $\phi$ but with all nuisance parameters substituted by consistent estimates.

Estimating the nuisance parameters on an independent sample can be attained via sample splitting \citep{Zheng2011} or cross-fitting \citep{chernozhukov2018}. In particular, the data can be split into $K$ disjoint subsets $I_k, \, k=1,\ldots,K$. Let $I_k^C$ denote the complement of the subset $I_k$. For each $k$, one can train data-adaptive nuisance parameter estimators using only the data in $I_k^C$. In calculating $\hat{\Psi}_\tau$, when a subject $i$ belongs to subset $I_k$, then the nuisance parameter estimators trained on $I_k^C$ should be used to plug in the nuisance parameter estimates in (\ref{DML}). When flexible data-adaptive estimates are used for the nuisance parameters, then this cross-fitting is necessary for the asymptotic validity of the proposal.

The advantage of utilizing the efficient influence function as the foundation for the estimator is that it exhibits consistent asymptotic behavior, regardless of whether it is based on the known nuisance parameters or consistent estimates thereof. This is contingent on the convergence rate conditions detailed in the following theorem. A proof is given in Appendix \ref{appB}.
\begin{theorem}\label{theorem2}
    The proposed estimator $\hat{\Psi}_\tau$ is asymptotically linear when
    \begin{enumerate}
        \item all nuisance parameter estimators are consistent,
        \item all nuisance parameters are estimated on a sample independent to the one on which $\hat{\Psi}_\tau$ is evaluated,
        \item all of the following terms are $o_p(n^{-1/2})$:
        \begin{align}
            &E\left[ \{\hat{Q}_\tau(Y\mid A,L)-Q_\tau(Y\mid A,L)\}^2 \right],\label{th2c1}\\ 
            &E\left( \left[1-\frac{f\{Q_\tau(Y\mid A,L)\mid A,L\}}{\hat{f}\{\hat{Q}_\tau(Y\mid A,L)\mid A,L\}}\right]^2 \right)^{1/2} E\left[ \{\hat{Q}_\tau(Y\mid A,L)-Q_\tau(Y\mid A,L)\}^2 \right]^{1/2},\label{th2c2}\\
            &E\left[\{E(A\mid L)-\hat{E}(A\mid L)\}^2\right]^{1/2} E\left( \left[E(Q_\tau(Y\mid A,L)\mid L)-\hat{E}\{\hat{Q}_\tau(Y\mid A,L)\mid L\}\right]^2 \right)^{1/2},\nonumber\\
            &E\left[\{E(A\mid L)-\hat{E}(A\mid L)\}^2\right],\nonumber
        \end{align}
        where the last condition is only needed when $\Psi_\tau\neq 0$, and where the expectations are conditional on the training sample and thus only over the distribution of $Y$, $A$ and $L$,
        \item $A-\hat{E}(A\mid L) = O_p(1)$,
        \item and $\;E[\{A-E(A\mid L)\}^2]>\epsilon$, $\;\hat{E}[\{A-\hat{E}(A\mid L)\}^2]>\epsilon\;$ and $\;\hat{f}\{\hat{Q}_\tau(Y\mid A,L)\mid A,L\}>\epsilon$ with probability 1 for some $\epsilon>0$.
    \end{enumerate}
\end{theorem}

The rate conditions stated in Theorem 2 will be satisfied e.g.\ if all nuisance parameters are consistently estimated at a rate faster than $n^{1/4}$, but are slightly less restrictive because of the rate conditions on product terms. The third condition in Theorem 2 states that (\ref{th2c1}) needs to be $o_p(n^{-1/2})$. However, if a faster convergence rate is achieved, then a slower convergence rate for $\hat{f}\{\hat{Q}_\tau(Y\mid A,L)\mid A,L\}$ may suffice, as long as the product (\ref{th2c2}) converges sufficiently fast. Similarly, if $E(A\mid L)$ can be estimated at a rate much faster than $n^{1/4}$, $E\{Q_\tau(Y\mid A,L)\mid L\}$ can be estimated at a slower rate. Under particular smoothness/sparsity assumptions, many data-adaptive methods can achieve these convergence rates (a summary is given in \cite{chernozhukov2018}).

The proposed estimator thus mitigates the plug-in bias that arises from the bias-variance trade-off in flexible data-adaptive techniques. These techniques are usually fine-tuned to minimize prediction error, leading to plug-in bias. It is shown that (see e.g.\ \cite{vdv2000asymptotic}) this plug-in bias can be approximated as minus the sample average of the influence functions. Since our estimator (\ref{DML}) is defined as the value of $\Psi_\tau$ that makes this sample average zero, the plug-in bias will be approximately zero by construction. Our proposal further takes into account the excess uncertainty induced by the use of data-adaptive procedures (e.g.\ variable selection or machine learning algorithms). By removing the bias induced by these procedures, our estimator has the asymptotic behaviour of an oracle estimator where the nuisance parameters are known. In this way, the proposed standard errors are honest in that they reflect the uncertainty that arises from these data-driven model choices.

\subsection{Targeted Learning}\label{sec: tmle}
Like other debiased machine learning strategies, targeted learning also removes plug-in bias based on the efficient influence function, but it does so after first `targeting' the initial nuisance parameter estimates \citep{vanderLaanRubin+2006, vanderlaanrose2011}. It thereby delivers an estimator that is asymptotically equivalent to the earlier defined debiased estimator under the conditions of theorem \ref{theorem2}. However, these estimators may have different finite sample performance \citep{yadlowsky2022}. Especially the last term in (\ref{DML}) may sometimes be unstable because of the density in the denominator. Therefore, we define a retargeted estimator $\widetilde{Q}_\tau(Y\mid A,L)$ which forces
\begin{align}\label{tmlefun}
    \frac{1}{n}\sum_{i=1}^n\left\{A_i-\hat{E}(A_i\mid L_i)\right\}\left[ \frac{\tau-I\{Y_i\leq \widetilde{Q}_\tau(Y_i\mid A_i,L_i)\}}{\hat{f}\{\widetilde{Q}_\tau(Y_i\mid A_i,L_i)\mid A_i,L_i\}} \right] = 0.
\end{align}
This is justified because (\ref{tmlefun}) is an unbiased estimating equation for $Q_\tau(Y\mid A,L)$. If (\ref{tmlefun}) holds and $\hat{Q}_\tau$ is replaced by $\widetilde{Q}_\tau$ in (\ref{DML}), the resulting estimator for $\Psi_\tau$ equals
\begin{align*}
    \frac{1}{n}\sum_{i=1}^n\frac{A_i-\hat{E}(A_i\mid L_i)}{\frac{1}{n}\sum_{i=1}^n\{A_i-\hat{E}(A_i\mid L_i)\}^2}\left[\widetilde{Q}_\tau(Y_i\mid A_i,L_i)- \hat{E}\{\widetilde{Q}_\tau(Y_i\mid A_i,L_i)\mid   L_i\} \right],
\end{align*}
which is a plug-in estimator of estimand (\ref{estimand3}).

Constructing an estimator that satisfies (\ref{tmlefun}) turns out to be rather challenging. For convenience, we will write
\begin{align}\label{weight}
    w(Q_{\tau,i}) = \frac{A_i-\hat{E}(A_i\mid L_i)}{\hat{f}\{Q_\tau(Y_i\mid A_i,L_i)\mid A_i,L_i\}}.
\end{align}
Using the weights defined in (\ref{weight}), we can write equation (\ref{tmlefun}) as
\begin{align*}
    \frac{1}{n}\sum_{i=1}^n w(\widetilde{Q}_{\tau,i})\left[\tau-I\{Y_i\leq \widetilde{Q}_\tau(Y_i\mid A_i,L_i)\} \right] = 0.
\end{align*}
This equation needs to be solved for $\widetilde{Q}_\tau(Y_i\mid A_i,L_i)$. Since it is highly nonlinear, we can only find an approximate solution by updating the initial estimator $\hat{Q}_\tau(Y_i\mid A_i,L_i)$ (e.g.\ obtained by using quantile regression forests) as follows:
\begin{enumerate}
    \item Let $Q_{\tau,i}^{(0)} = \hat{Q}_\tau(Y_i\mid A_i,L_i)$ and $\epsilon_0=0$.
    \item For $k>0$: let $Q_{\tau,i}^{(k)} = Q_{\tau,i}^{(k-1)} + \epsilon_{k-1} w\left(Q_{\tau,i}^{(k-1)}\right)$, with
    \[
    \epsilon_{k-1} = \argmin_\epsilon \left\lvert  \frac{1}{n}\sum_{i=1}^n w\left(Q_{\tau,i}^{(k-1)}\right)\left[\tau-I\left\{Y_i\leq Q_{\tau,i}^{(k-1)}+\epsilon w\left(Q_{\tau,i}^{(k-1)}\right)\right\} \right] \right\rvert .
    \]
    \item Repeat step 2 until the value of
\[
    \left\lvert  \frac{1}{n}\sum_{i=1}^n w\left(Q_{\tau,i}^{(k-1)}\right)\left[\tau-I\left\{Y_i\leq Q_{\tau,i}^{(k-1)}+\epsilon w\left(Q_{\tau,i}^{(k-1)}\right)\right\} \right] \right\rvert 
\]
is no longer decreasing or achieves a certain threshold close to 0.
\end{enumerate}
The resulting estimator for $Q_\tau(Y_i\mid A_i,L_i)$ is denoted $\widetilde{Q}_\tau(Y_i\mid A_i,L_i)$. By now substituting $\hat{Q}_\tau(Y_i\mid A_i,L_i)$ with $\widetilde{Q}_\tau(Y_i\mid A_i,L_i)$ in the estimating equations estimator (\ref{DML}), we obtain a Targeted Maximum Likelihood Estimation (TMLE) estimator
\begin{align}\label{TMLE}
\begin{split}
    \hat{\Psi}_\tau^{TMLE} &= \frac{1}{n}\sum_{i=1}^n\frac{A_i-\hat{E}(A_i\mid L_i)}{\frac{1}{n}\sum_{i=1}^n\{A_i-\hat{E}(A_i\mid L_i)\}^2}\left[\widetilde{Q}_\tau(Y_i\mid A_i,L_i)\right.\\
    &\left.\qquad\qquad- \hat{E}\{\widetilde{Q}_\tau(Y_i\mid A_i,L_i)\mid   L_i\} + \frac{\tau-I\{Y_i\leq \widetilde{Q}_\tau(Y_i\mid A_i,L_i)\}}{\hat{f}\{\widetilde{Q}_\tau(Y_i\mid A_i,L_i)\mid A_i,L_i\}} \right].
\end{split}
\end{align}
The targeting step results in the last term of (\ref{TMLE}) being approximately zero, such that
\begin{align}\label{TMLE2}
    \hat{\Psi}_\tau^{TMLE} \approx \frac{1}{n}\sum_{i=1}^n\frac{A_i-\hat{E}(A_i\mid L_i)}{\frac{1}{n}\sum_{i=1}^n\{A_i-\hat{E}(A_i\mid L_i)\}^2}\left[\widetilde{Q}_\tau(Y_i\mid A_i,L_i)- \hat{E}\{\widetilde{Q}_\tau(Y_i\mid A_i,L_i)\mid   L_i\} \right].
\end{align}
Comparing (\ref{TMLE2}) with (\ref{plugin}), it is clear that we removed (approximately) all bias from the initial plug-in estimator by cleverly targeting the nuisance parameter estimator for $Q_\tau(Y_i\mid A_i,L_i)$. In Appendix \ref{appB}, we give a sketch of a proof that (\ref{TMLE}) and (\ref{TMLE2}) are asymptotically equivalent, even with only one iteration of targeting. Nevertheless, we opt to use (\ref{TMLE}) to mitigate reliance on the accuracy of approximation (\ref{TMLE2}).

The targeted learning estimator with 5-fold cross-fitting can be obtained as follows. After splitting the data into 5 equal parts $I_k$, with $k=1,\ldots,5$, five models are trained for each nuisance parameter, where each time one subset $I_k$ is left out from the data. In particular, a nonparametric model is trained for $E(A\mid L)$, $Q_\tau(Y\mid A,L)$, $E\{Q_\tau(Y\mid A,L)\mid L\}$ and $f\{Q_\tau(Y\mid A,L)\mid A,L\}$ on $I_k^C$ (i.e.\ 80\% of the data). Then, those models are used to estimate the nuisance parameters for subjects belonging to $I_k$. These estimates are then used to calculate the contribution to $\Psi_\tau$ (which is a sample average) for the same subjects. The targeting algorithm described above is then applied to the entire data set. Then, $\Psi_\tau$ is estimated as in (\ref{TMLE}). The standard error can be estimated similarly to the debiased estimator (\ref{DML}), i.e.\ as the square root of 
\begin{align*}
    &\frac{1}{n^2}\sum_{i=1}^n\left(\frac{A_i-\hat{E}(A_i\mid L_i)}{\frac{1}{n}\sum_{i=1}^n\{A_i-\hat{E}(A_i\mid L_i)\}^2}\left[\widetilde{Q}_\tau(Y_i\mid A_i,L_i)- \hat{E}\{\widetilde{Q}_\tau(Y_i\mid A_i,L_i)\mid   L_i\}\right.\right.\\
    &\left.\left.\qquad\qquad + \frac{\tau-I\{Y_i\leq \widetilde{Q}_\tau(Y_i\mid A_i,L_i)\}}{\hat{f}\{\widetilde{Q}_\tau(Y_i\mid A_i,L_i)\mid A_i,L_i\}}- \hat{\Psi}_\tau^{TMLE} \{A_i-E(A_i\mid L_i)\} \right]\right)^2.
\end{align*}

To estimate the nuisance parameter $E\{Q_\tau(Y\mid A,L)\mid L\}$, one would first need to estimate $Q_\tau(Y\mid A,L)$. Then, based on these estimates, one can estimate $E\{Q_\tau(Y\mid A,L)\mid L\}$ by taking the quantile estimates as the outcome for a model that predicts its conditional expectation given $L$. However, when the exposure $A$ is binary, this expectation reduces to
\begin{equation}\label{decomposition}
    E\{Q_\tau(Y\mid A,L)\mid L\} = Q_\tau(Y\mid A=1,L)E(A\mid L) + Q_\tau(Y\mid A=0,L)\{1-E(A\mid L)\}.
\end{equation}
This allows us to estimate this nuisance parameter without additionally modeling the conditional expectation of $Q_\tau(Y\mid A,L)$ given $L$. Indeed, $Q_\tau(Y\mid A,L)$ and $E(A\mid L)$ are nuisance parameters themselves so they need to be modeled anyway. These models can then be used to estimate $Q_\tau(Y\mid A=1,L)$, $Q_\tau(Y\mid A=0,L)$ and $E(A\mid L)$. Furthermore, targeted learning with cross-fitting can easily be applied in this case, as $$\hat{E}\{\widetilde{Q}_\tau(Y\mid A,L)\mid L\} = \widetilde{Q}_\tau(Y\mid A=1,L)\hat{E}(A\mid L) + \widetilde{Q}_\tau(Y\mid A=0,L)\{1-\hat{E}(A\mid L)\}.$$

On the contrary, decomposition (\ref{decomposition}) cannot be made for a continuous exposure. It is then necessary to have an additional model for $E\{Q_\tau(Y\mid A,L)\mid L\}$. However, since the retargeting of $\hat{Q}_\tau(Y\mid A,L)$ is done on the full data set, it is no longer possible to do proper cross-fitting for $E\{Q_\tau(Y\mid A,L)\mid L\}$. In that case, a model to estimate the conditional expectation of $\widetilde{Q}_\tau(Y\mid A,L)$ given $L$ should be trained using only 80\% of the data, but $\widetilde{Q}_\tau(Y\mid A,L)$ would already have been targeted using all data. This makes it no longer possible to train a model for $E\{\widetilde{Q}_\tau(Y\mid A,L)\mid L\}$ without involving the remaining 20\% of the data. 
One way to resolve this problem is to do the targeting in each fold separately. However, forcing equation (\ref{tmlefun}) to hold in every fold does not guarantee that it holds for the entire data set, especially in smaller samples. Alternatively, one could choose not to target the quantile estimator in $E\{Q_\tau(Y\mid A,L)\mid L\}$. However, instead of repeating the above targeting algorithm until convergence, we propose to do step 2 only once, which is justified by the sketched proof in Appendix \ref{appB}. Thus, the resulting retargeted estimator for $Q_\tau(Y\mid A,L)$ will be
\begin{equation*}
    \widetilde{Q}_\tau(Y\mid A,L) = \hat{Q}_\tau(Y\mid A,L) + \hat{\epsilon}  w(\hat{Q}_{\tau,i}),
\end{equation*}
with 
\begin{equation}\label{eps onestep}
\hat{\epsilon} = \argmin_\epsilon \left\lvert  \frac{1}{n}\sum_{i=1}^n w\big(\hat{Q}_{\tau,i}\big)\left[\tau-I\left\{Y_i\leq \hat{Q}_{\tau,i}+\epsilon w\big(\hat{Q}_{\tau,i}\big)\right\} \right] \right\rvert .
\end{equation}
The nuisance parameter $E\{Q_\tau(Y\mid A,L)\mid L\}$ can then be estimated based on the following identity:
\begin{align}\label{epsmodels}
\begin{split}
    E\{\widetilde{Q}_\tau(Y\mid A,L)\mid L\} &= E\{\hat{Q}_\tau(Y\mid A,L) + \hat{\epsilon} w(\hat{Q}_{\tau,i})\mid  L\}\\
    &= E\{\hat{Q}_\tau(Y\mid A,L)\mid L\} + \hat{\epsilon} E\left[\frac{A-\hat{E}(A\mid L)}{\hat{f}\{\hat{Q}_\tau(Y\mid A,L)\mid A,L\}} \mid  L\right].
\end{split}
\end{align}
Therefore, an estimator can be constructed by estimating $E\{\hat{Q}_\tau(Y\mid A,L)\mid L\}$ and $E[\{A-\hat{E}(A\mid L)\}/\hat{f}\{Q_\tau(Y\mid A,L)\mid A,L\}\mid L]$ separately, possibly via cross-fitting. In this way, the value of $\epsilon$ can be learned via (\ref{eps onestep}) after estimating the 2 models in (\ref{epsmodels}) that determine $E\{Q_\tau(Y\mid A,L)\mid L\}$, which allows us to do the retargeting on the entire data set.

\subsection{Parametric quantile regression with variable selection}\label{simplVS}
When using quantile regression, possibly with (stepwise) variable selection, to estimate the nuisance parameter $Q_\tau(Y\mid A,L)$, the debiased estimator (\ref{DML}) and TMLE estimator (\ref{TMLE}) can be further simplified. In particular, assume that $Q_\tau(Y\mid A,L) = \beta_\tau A+ \alpha^T L^*$, with $L^*$ a subset of variables of $L$. This simplifies the efficient influence function of $\Psi_\tau$ to
\begin{align*}
    & \frac{A-E(A\mid L)}{E[\{A-E(A\mid L)\}^2]}\left[\beta \{A - E(A\mid  L)\} + \frac{\tau-I\{Y\leq Q_\tau(Y\mid A,L)\}}{f\{Q_\tau(Y\mid A,L)\mid A,L\}} - \Psi_\tau \{A-E(A\mid L)\} \right].
\end{align*}
The value of $\Psi_\tau$ that makes the sample average of the influence functions zero is then
\begin{align}\label{DMLvs}
\begin{split}
    \hat{\Psi}_\tau &= \hat{\beta}_\tau + \frac{1}{n}\sum_{i=1}^n\frac{A_i-\hat{E}(A_i\mid L_i)}{\frac{1}{n}\sum_{i=1}^n\{A_i-\hat{E}(A_i\mid L_i)\}^2}\left[ \frac{\tau-I\{Y_i\leq \hat{Q}_\tau(Y_i\mid A_i,L_i)\}}{\hat{f}\{\hat{Q}_\tau(Y_i\mid A_i,L_i)\mid A_i,L_i\}} \right].
\end{split}
\end{align}
Here, $\hat{\beta}_\tau$ is the estimate of $\beta_\tau$ obtained by fitting a quantile regression model of $Y$ on $A$ and (a subset of) $L$ via standard software (e.g.\ the R package `quantreg'), possibly with variable selection. The second term eliminates bias due to variable selection mistakes, which is needed to achieve valid post-selection inference.

An asymptotically equivalent (one-step) TMLE estimator is
\begin{align}\label{TMLEvs}
\begin{split}
    \hat{\Psi}_\tau^{TMLE} &= \hat{\beta} + \frac{1}{n}\sum_{i=1}^n\frac{A_i-\hat{E}(A_i\mid L_i)}{\frac{1}{n}\sum_{i=1}^n\{A_i-\hat{E}(A_i\mid L_i)\}^2}\left[\left(\frac{A_i-\hat{E}(A_i\mid L_i)}{\hat{f}\{\widetilde{Q}_\tau(Y_i\mid A_i,L_i)\mid A_i,L_i\}}\right.\right.\\
    &\left.\left.\qquad-\hat{E}\left[\frac{A_i-\hat{E}(A_i\mid L_i)}{\hat{f}\{\widetilde{Q}_\tau(Y_i\mid A_i,L_i)\mid A_i,L_i\}}\mid L_i\right]\right)\hat{\epsilon}
    + \frac{\tau-I\{Y_i\leq \widetilde{Q}_\tau(Y_i\mid A_i,L_i)\}}{\hat{f}\{\widetilde{Q}_\tau(Y_i\mid A_i,L_i)\mid A_i,L_i\}} \right],
\end{split}
\end{align}
where the last term is approximately zero. More details on the construction of this estimator can be found in Appendix \ref{appC}.

\section{What if model (\ref{model}) is wrong?}\label{sec:estimand}
Model (\ref{model}) assumes a linear association between $A$ and the quantiles of $Y$ conditional on $L$, and the absence of $A$-$L$ interactions. Such assumptions are often made for parsimony and better interpretability, for example, to be able to summarize the association of interest into one single number. The downside is that the model may well fail to represent the underlying data-generating mechanism. 
Instead of just fitting this model and drawing inference for $\beta_\tau$ as one would do for a model parameter in standard quantile regression, we have started our analysis by defining a model-free estimand which reduces to $\beta_\tau$ if model (\ref{model}) is correct, but that still captures the association of interest if it is misspecified. The latter would, for instance, be satisfied when the estimand equals some $L$-dependent weighted average of the estimand which one would report for a subset of individuals with given $L$ (e.g.\ in the case of a dichotomous exposure, the difference in $\tau$-quantile of the outcome between individuals with $A=1$ versus $A=0$ and the same $L$). Building on \cite{AssumptionLean}, we will focus on the estimand (\ref{estimand2}) or its equivalent representation (\ref{estimand3}).
For a dichotomous exposure that is independent of $L$ (e.g.\ when $A$ is a binary randomized treatment) this estimand simplifies further to:
\begin{align}\label{estimandrandbin}
    \Psi_\tau=E\left\{Q_\tau(Y\mid A=1,L)-Q_\tau(Y\mid A=0,L)\right\},
\end{align}
which equals an expected difference in $\tau$-quantile of the outcome between individuals with $A=1$ versus $A=0$ and the same level of $L$. The proposed estimators thus remain interpretable in terms of average differences between quantiles, even when model (\ref{model}) is misspecified. This is in stark contrast to (partially linear) quantile regression that does not provide such a guarantee.

With interest in a causal relationship, one can make use of counterfactual outcomes $Y^a$ to denote the outcome that would have been seen, if the individual would have had exposure value $A=a$. Under consistency (i.e.\ $Y^a=Y$) and exchangeability (i.e.\ when $A$ is independent of the counterfactual outcome $Y^a$ to exposure level $a$, $Y^a\indep A\mid L$), the estimand (\ref{estimandrandbin}) reduces to
\begin{align*}
    \Psi_\tau=E\left\{Q_\tau(Y^1\mid L)-Q_\tau(Y^0\mid L)\right\},
\end{align*}
which is an expected difference in $\tau$-quantile of the counterfactual outcome $Y^1$ versus $Y^0$ conditional on $L$.

If $A$ is dichotomous but not randomized, then the estimand reduces to
\begin{align}\label{estimandbin}
    \Psi_\tau=\frac{E[\pi(L)\{1-\pi(L)\}\left\{Q_\tau(Y\mid A=1,L)-Q_\tau(Y\mid A=0,L)\right\}]}{E[\pi(L)\{1-\pi(L)\}]},
\end{align}
with $\pi(L)=\text{pr}(A=1\mid L)$, the propensity score. This is a weighted average of the difference in $\tau$-quantile of the outcome between individuals with $A=1$ versus $A=0$ and the same level of $L$, with weights equal to 
\begin{equation*}\label{weights}
    \frac{\pi(L)\{1-\pi(L)\}}{E[\pi(L)\{1-\pi(L)\}]}.
\end{equation*}
This choice of weights is motivated by settings where interest lies in the causal effect of $A$. One then wishes to upweight those strata where $\Var(A\mid L)=\pi(L)\{1-\pi(L)\}$ is relatively large. This corresponds to giving the highest weight to covariate regions where both exposed and unexposed subjects are found. Moreover, in strata of $L$ where all subjects are exposed (or unexposed), those subjects will not contribute to $\Psi_\tau$. Those strata do not contain pertinent information about the causal effect of $A$ and will be given a weight of zero.

As above, under consistency and conditional exchangeability, (\ref{estimandbin}) reduces to a weighted average of the difference in $\tau$-quantile of $Y^1$ versus $Y^0$ conditional on $L$:
\begin{align*}
    \Psi_\tau=\frac{E\left[\pi(L)\{1-\pi(L)\}\left\{Q_\tau(Y^1\mid L)-Q_\tau(Y^0\mid L)\right\}\right]}{E[\pi(L)\{1-\pi(L)\}]}.
\end{align*}

For exposures that are not necessarily binary, this estimand can be generalized by (\ref{estimand2}), which is constructed as follows. Consider drawing a sample of individuals, and for each, randomly drawing an individual with the same $L$. A population least squares regression of the difference between their conditional quantiles on their difference in exposure values results in (\ref{estimand2}), or equivalently, in (\ref{estimand3}). This is closely related to the idea of matching \citep{rosenbaum1983,rosenbaum1989} based on $L$. 

Further insight can be obtained by considering the case where
\[
    Q_\tau(Y\mid A,L) = \beta_\tau(L) A + \omega_\tau(L),
\]
where the coefficient $\beta_\tau$ is now allowed to be $L$-dependent. By substituting the right-hand side in (\ref{estimand3}), the estimand can easily be shown to reduce to
\[
\Psi_\tau = \frac{E\left\{\Var(A\mid L)\beta_\tau(L)\right\}}{E\left\{\Var(A\mid L)\right\}}.
\]
This is a weighted average of the $L$-dependent coefficients $\beta_\tau(L)$, which thus captures the conditional association of interest.

\section{Simulation experiments}
To investigate the performance of the proposed estimators, we use four simulation experiments. In all experiments, we report Monte Carlo bias and standard deviation (SD), as well as the averaged estimated standard error (SE) and coverage of corresponding 95\% Wald confidence intervals. For all experiments, we show results for 1000 simulation runs for the oracle estimator (Oracle), obtained by fitting a correctly specified standard quantile regression model that includes the exposure $A$ and all necessary confounders, the naive plug-in (Plug-in) estimator (\ref{plugin}), the DML estimator (\ref{DML}) based on quantile random forests without cross-fitting (DML) and with 5-fold cross-fitting (DML-CF), and the targeted learning estimator (\ref{TMLE}) without cross-fitting (TMLE) and with 5-fold cross-fitting (TMLE-CF). The sample size is fixed at $n=500$ unless otherwise stated. The standard error of the plug-in estimator is always calculated similarly as for the debiased machine learning estimator, but removing the last term of (\ref{DML}) from the efficient influence function. More experiment-specific details are provided below.

\subsection{Experiment 1: Homoscedasticity, heteroscedasticity, and continuity}
This experiment consisted of 3 simulation settings, where the data generating mechanism was each time slightly changed. We studied the proposed estimators' behaviour for a binary exposure in the first simulation setting, inspired by a setting used in \cite{Sun2021}. Some minor adjustments were made to prevent the coefficient of determination ($R^2$) from exceeding 0.8, as this is often unrealistic. In setting 1, we considered a homoscedastic outcome where the exposure effect had the same magnitude for all quantiles. Thus, there was a distribution shift between exposed and unexposed individuals. In setting 2, we adjusted the setting slightly so that the outcome variability differed for both exposure groups. This implies that the exposure effect is different for different values of $\tau$. In setting 3, the exposure was continuous. The results are shown in Table \ref{table1}.

In the first setting, the data was generated as follows. Four confounding variables were generated according to
\begin{align}\label{confounders}
\begin{pmatrix}
    L_1\\
    L_2\\
    L_3\\
    L_4
\end{pmatrix} 
\sim \mathcal{N} \left\{ 
\begin{pmatrix}
    0\\
    0\\
    0\\
    0
\end{pmatrix} ,
\begin{pmatrix}
    1 & 0.5 & 0.2 & 0.3\\
    0.5 & 1 & 0.7 & 0\\
    0.2 & 0.7 & 1 & 0\\
    0.3 & 0 & 0 & 1
\end{pmatrix}
\right\}.
\end{align}
The dichotomous exposure $A$ follows a Bernoulli distribution conditional on $L=(L_1,L_2,L_3,L_4)$, with true propensity score given by $\text{pr}(A=1\mid L) = \text{expit}(-0.5+0.2 L_1 - 0.4 L_2 -0.4 L_3+0.2 L_4)$.

In the homoscedastic setting, the outcome $Y$ was generated as 
\begin{align}\label{outcome}
Y = 1+A+\sin(L_1)+L_2^2+L_3+L_4+L_3 L_4+\epsilon,
\end{align}
where $\epsilon$ follows a Gamma($k,\theta$) distribution with shape parameter $k=1$ and scale parameter $\theta=2$ and is independent from $A$ and $L$. Under heteroscedasticity, we chose $\epsilon\sim\text{Gamma}(1,2+A)$.

For all estimators except Oracle, quantile random forests (via the `grf' package described in \cite{athey2019generalized}) were used to learn $Q_\tau(Y\mid A,L)$. To learn $E(A\mid L)$, SuperLearner was used; this is an ensemble learner described in \cite{vanderLaanSL2007}. The SuperLearner library included generalized linear models with and without elastic net regularization, generalized additive models, and random forests. These estimators were also used to estimate $E\{Q_\tau(Y\mid A,L)\mid L\}$ based on (\ref{decomposition}). To obtain the nuisance parameter $f\{Q_\tau(Y\mid A,L)\mid A,L\}$, we estimated the conditional density of $Y$ given $A$ and $L$ via the `fk\_density' function from the R package `FKSUM' \citep{fksum1,fksum2}. This function gives a fast approximation of a univariate density based on recursive computations in evaluation points of choice. Since our interest lies in estimating a conditional density function, we estimated the density of $Y-Q_\tau(Y\mid A,L)$ and evaluated it at 0. This estimation procedure assumes that the distribution of $Y-Q_\tau(Y\mid A,L)$ is independent of $A$ and $L$. 

\begin{table}[h]
\centering 
\addtolength{\leftskip} {-4cm}
\addtolength{\rightskip}{-4cm}
\begin{tabular}{clcccccccccccccc}
 &  &   \multicolumn{4}{c}{$\tau = 0.5$} &  & \multicolumn{4}{c}{$\tau = 0.75$} &  & \multicolumn{4}{c}{$\tau = 0.9$} \\ 
\multirow{-2}{*}{\textbf{Setting}} & \multirow{-2}{*}{\textbf{estimator}}  & \textbf{bias} & \textbf{SD} & \textbf{SE} & \textbf{Cov} &  & \textbf{bias} & \textbf{SD} & \textbf{SE} & \textbf{Cov} &  & \textbf{bias} & \textbf{SD} & \textbf{SE} & \textbf{Cov} \\ 
 & Oracle   & -0.17 & 19 & 20 & 96.6 &  & -0.51 & 33 & 35 & 96.3 &  & -1.1 & 56 & 60 & 96.0 \\
 & Plug-in   & -70 & 12 & 1.5 & 0.1 &  & -72 & 14 & 1.8 & 0.5 &  & -64 & 22 & 3.6 & 1.6 \\
 & DML   & -15 & 19 & 19 & 85.9 &  & -35 & 29 & 25 & 63.7 &  & -56 & 32 & 22 & 33.0 \\
 & DML-CF   & -16 & 18 & 18 & 84.2 &  & -24 & 28 & 25 & 79.2 &  & -31 & 43 & 35 & 77.4 \\
 & TMLE   & -41 & 21 & 19 & 45.2 &  & -71 & 32 & 26 & 27.0 &  & -110 & 28 & 45 & 12.6 \\
 \multirow{-6}{*}{1} & TMLE-CF &  1.2 & 22 & 25 & 97.2 &  & 2.8 & 39 & 37 & 93.5 &  & 14 & 68 & 63 & 91.4 \\
 & Oracle   & -0.59 & 26 & 26 & 93.7 &  & -2.7 & 45 & 46 & 93.6 &  & -6.8 & 75 & 78 & 94.4 \\
 & Plug-in   & -89 & 26 & 2.7 & 0.6 &  & -120 & 42 & 4.1 & 1.1 &  & -160 & 60 & 6.8 & 1.5 \\
 & DML & -15 & 27 & 22 & 83.0 &  & -45 & 46 & 30 & 57.8 &  & -120 & 66 & 28 & 16.5 \\
 & DML-CF   & -20 & 26 & 20 & 75.6 &  & -43 & 43 & 30 & 61.7 &  & -78 & 69 & 42 & 49.6 \\
 & TMLE   & -45 & 27 & 23 & 49.4 &  & -76 & 49 & 32 & 38.0 &  & -195 & 67 & 34 & 4.0 \\
\multirow{-6}{*}{2} & TMLE-CF & 5.0 & 29 & 28 & 93.9 &  & 10 & 50 & 42 & 88.7 &  & 13 & 88 & 70 & 87.6 \\
 & Oracle   & -0.13 & 3.5 & 3.6 & 95.6 &  & -0.29 & 5.8 & 6.3 & 95.8 &  & 0.10 & 10 & 11 & 94.6 \\
 & Plug-in   & -17 & 6.4 & 1.6 & 0.5 &  & -25 & 7.9 & 1.9 & 0.2 &  & -39 & 11 & 2.1 & 0.0 \\
 & DML   & -7.9 & 6.8 & 3.6 & 45.1 &  & -16 & 8.5 & 4.6 & 21.6&  & -34 & 11 & 4.3 & 1.7 \\
 & DML-CF   & -6.2 & 3.9 & 3.1 & 48.5 &  & -10 & 6.0 & 4.4 & 40.0 &  & -19 & 9.1 & 6.2 & 23.0\\
 & TMLE   & -11 & 7.1 & 3.6 & 29.3 &  & -18 & 9.3 & 4.9 & 18.9 &  & -41 & 12 & 5.2 & 1.6 \\
\multirow{-6}{*}{3} & TMLE-CF & -1.1 & 4.4 & 4.2 & 92.9 &  & 0.016 & 7.7 & 6.3 & 90.0 &  & 1.2 & 14 & 10 & 85.3
\end{tabular}
\caption{Simulation results for experiment 1: sample size $n=500$, quantile $\tau$, Monte Carlo bias, Monte Carlo standard deviation (SD), average of the influence function based standard errors (SE) and coverage of 95\% Wald confidence intervals (Cov). All values have been multiplied by $10^2$.}
\label{table1}
\end{table}

Estimating the conditional density could alternatively be done via the `haldensify' package \citep{hejazi2022efficient,hejazi2022haldensify-rpkg,hejazi2022haldensify-joss}, which uses the highly adaptive lasso (HAL) \citep{HAL2016,vanderLaanHAL2017}. This conditional density could then be evaluated in the estimated quantile without making any additional assumptions. However, this method is much more time-consuming and did not outperform the simpler estimation method. The corresponding results can be found in Appendix \ref{appD}.

In the third setting of this experiment, we studied the proposed estimators' finite sample behaviour for a continuous exposure. Here, $A$ was generated according to a normal distribution $\mathcal{N}(\mu,2^2)$, with $\mu = -0.5+L_1 - 2 L_2 -2 L_3+L_4$.
The outcome $Y$ was generated from (\ref{outcome}), where $\epsilon\sim\text{Gamma}(1,4)$ is independent from $A$ and $L$. The nuisance parameters were all estimated as before, except for $E\{Q_\tau(Y\mid A,L)\mid L\}$ as the decomposition (\ref{decomposition}) is no longer possible for continuous $A$. To estimate this parameter, we again used SuperLearner with the same library as before, using the random forest estimate of $Q_\tau(Y\mid A,L)$ as outcome. For the TMLE-CF estimator, the nuisance parameters described at the end of Section \ref{sec: tmle} are also estimated using SuperLearner.

The simulations all demonstrate the poor performance of the naive plug-in estimator. It has a high bias and the standard error is poorly estimated and overly optimistic, leading to a coverage of almost 0. When comparing the debiased machine learning and targeted learning estimators with versus without cross-fitting, the advantage of cross-fitting becomes clearly visible. In general, the estimators with cross-fitting perform better, especially for $\tau=0.9$ and in the setting with continuous exposure. The use of cross-fitting leads to a bias reduction and to estimated standard errors that better approximate the Monte Carlo standard deviation, resulting in a better coverage. For $\tau=0.5$, we sometimes see the estimators without cross-fitting performing slightly better in terms of bias, likely because the nuisance parameters are estimated on larger sample sizes. When comparing DML-CF with TMLE-CF, we often observe the latter having a larger SD. However, in terms of bias and coverage, we always observe targeted learning outperforming the other estimator. As expected, all estimators perform worse for more extreme quantiles. 

\subsection{Experiment 2: Extreme propensity scores}
In this experiment, we generated a setting where the propensity scores could take extreme values (i.e.\ values close to 0 or 1), leading to little overlap between exposed and unexposed individuals. Therefore, we generated a dichotomous exposure $A$, with $\text{pr}(A=1\mid L) = \text{expit}(-0.5+0.2 L_1 - 0.4 L_2 -0.4 L_3+0.2 L_4 +0.5 L_1^2-0.5 L_2^2+0.5 L_3L_4)$,
where $L$ follows the same multivariate normal distribution (\ref{confounders}) as in experiment 1. To gain insight into these propensity scores, histograms are provided in Appendix \ref{appD}. The outcome was generated according to (\ref{outcome}), with $\epsilon\sim\text{Gamma}(1,3)$. We evaluated the same estimators as in experiment 1 and included a naive (misspecified) quantile regression model (QR) that only included main effects. The results are shown in Table \ref{tab:extrprop}. 

\begin{table}[h]
\centering 
\addtolength{\leftskip} {-4cm}
\addtolength{\rightskip}{-4cm}
\begin{tabular}{clcccccccccccccc}
 \textbf{Sample}&  &   \multicolumn{4}{c}{$\tau = 0.5$} &  & \multicolumn{4}{c}{$\tau = 0.75$} &  & \multicolumn{4}{c}{$\tau = 0.9$} \\  
\textbf{size} & \multirow{-2}{*}{\textbf{estimator}}  & \textbf{bias} & \textbf{SD} & \textbf{SE} & \textbf{Cov} &  & \textbf{bias} & \textbf{SD} & \textbf{SE} & \textbf{Cov} &  & \textbf{bias} & \textbf{SD} & \textbf{SE} & \textbf{Cov} \\
 & Oracle   & -0.30 & 29 & 30 & 94.6 &  & 0.68 & 51 & 52 & 94.8 &  & -3.0 & 92 & 91 & 94.2 \\
  & QR   & 16 & 36 & 36 & 92.3 &  & 33 & 55 & 58 & 92.7 &  & 44 & 91 & 95 & 94.1 \\
 & Plug-in   & -72 & 14 & 1.5 & 0.0 &  & -74 & 17 & 2.0 & 0.4 &  & -62 & 36 & 4.7 & 5.8 \\
 & DML   & -12 & 32 & 33 & 93.1 &  & -37 & 44 & 41 & 82.0 &  & -56 & 54 & 35 & 54.0 \\
 & DML-CF   & -21 & 27 & 25 & 85.3 &  & -35 & 40 & 36 & 80.5 &  & -39 & 63 & 46 & 74.7 \\
 & TMLE   & -59 & 33 & 33 & 55.0 &  & -97 & 45 & 43 & 35.1 &  & -110 & 97 & 45 & 31.8 \\
 \multirow{-6}{*}{$500$} & TMLE-CF &  -6.5 & 32 & 34 & 95.4 &  & -13 & 55 & 52 & 92.9 &  & -14 & 100 & 87 & 91.2 \\
 & Oracle   & -0.089 & 20 & 21 & 95.6 &  & 0.37 & 36 & 37 & 95.0 &  & -3.0 & 61 & 63 & 95.3 \\
  & QR   & 17 & 25 & 26 & 91.1 &  & 33 & 40 & 40 & 86.4 &  & 43 & 64 & 65 & 90.8 \\
 & Plug-in   & -66 & 12 & 1.1 & 0.0 &  & -69 & 16 & 1.5 & 0.2 &  & -60 & 27 & 3.2 & 2.1 \\
 & DML & -14 & 21 & 22 & 90.5 &  & -34 & 32 & 29 & 74.3 &  & -57 & 36 & 23 & 34.0 \\
 & DML-CF   & -22 & 19 & 17 & 73.6 &  & -32 & 29 & 24 & 68.1 &  & -41 & 43 & 29 & 59.7 \\
 & TMLE   & -59 & 22 & 22 & 24.1 &  & -100 & 33 & 30 & 9.7 &  & -130 & 86 & 29 & 17.2 \\
\multirow{-6}{*}{$1000$} & TMLE-CF & -6.8 & 22 & 23 & 94.3 &  & -9.5 & 39 & 37 & 93.0 &  & -14 & 69 & 62 & 90.7 
\end{tabular}
\caption{Simulation results for experiment 2: sample size $n$, quantile $\tau$, Monte Carlo bias, Monte Carlo standard deviation (SD), average of the influence function based standard errors (SE) and coverage of 95\% Wald confidence intervals (Cov). All values have been multiplied by $10^2$.}
\label{tab:extrprop}
\end{table}

The proposed estimators show similar behaviour in this experiment as in the first experiment. The best performance is again achieved by the TMLE-CF estimator. It clearly performs better than a naive quantile regression, particularly in terms of bias. When increasing the sample size from 500 to 1000, TMLE-CF does not significantly change performance. However, QR's bias does not decrease, unlike its standard error, leading to a decrease in coverage of 95\% confidence intervals. The TMLE estimator also has a large bias, which is likely caused by overfitting since the bias of TMLE-CF is much smaller. As in the first experiment, cross-fitting provides a major improvement.

\subsection{Experiment 3: Randomized dichotomous exposure}
In this experiment, we first chose the exposure to be randomized with $\text{pr}(A=1)=0.5$, independently from all other variables. Confounders follow the same multivariate normal distribution (\ref{confounders}) as in experiment 1. The outcome was generated from (\ref{outcome}), again with $\epsilon\sim\text{Gamma}(1,2)$. We compared our estimators with 2 estimators recently proposed in \cite{athey2021semiparametric}. Both estimators focus on comparing marginal quantiles $Q_\tau(Y^a)$ of counterfactual distributions $Y^a$, $a=0,1$. They are designed for a constant quantile treatment effect between 2 groups, meaning that there is a location shift between the outcome distribution of exposed versus unexposed subjects. \cite{athey2021semiparametric} estimate so-called weighted average quantile treatment effects
\begin{align}\label{waqte}
    \int_0^1\left\{ Q_\tau(Y^1)-Q_\tau(Y^0) \right\}\,\mathrm{d}W(\tau),
\end{align}
for certain weights $W(\tau)$ that integrate to one, and $W(0)=0$, $W(1)=1$. The first estimator (IFB) is an influence function based estimator, derived from the efficient influence function of (\ref{waqte}). It is a semiparametric efficient estimator for (\ref{waqte}). The second estimator (WAQ) is a weighted average quantile estimator. It is a plug-in estimator for (\ref{waqte}):
\begin{align*}
    \int_0^1\left\{ \hat{Q}_\tau(Y^1)-\hat{Q}_\tau(Y^0) \right\}w_f(\tau)\,\mathrm{d}\tau,
\end{align*}
where the weights $w_f(\tau)$ are such that they provide an efficient L-estimator under the location shift assumption and the assumption that the density $f$ of the counterfactual outcome is known (see \cite{athey2021semiparametric} for details). 
Since these estimators assume a location shift, they are only implemented for estimating the median. The results are shown in Table \ref{tab: exp3} (left).

For the proposed estimators we see similar results as in experiment 1. Targeted learning with cross-fitting clearly outperforms the other estimators with a bias of 0.013 and coverage of 95.1\%. Despite our proposal being assumption-lean, it still performs slightly better than the estimators defined in \cite{athey2021semiparametric}, despite all assumptions being met in this simulation experiment.

\begin{table}[h]
\centering
\begin{tabular}{cccccccccccc}
\textbf{estimator}   & \textbf{bias} & \textbf{SD} & \textbf{SE} & \textbf{Cov} &&&&&&& \\
Oracle & 0.0071 & 0.18 & 0.19 &  95.5 &&&&&&&    \\
Plug-in & -0.70 & 0.12 & 0.011 &  0.0  &&&&&&&  \\
DML & -0.16 & 0.20 & 0.18 &  80.3    &&\textbf{estimator}   & \textbf{bias} & \textbf{SD} & \textbf{SE} & \textbf{Cov}  \\
DML-CF & -0.18 & 0.19 & 0.17 &  76.1  &&TMLE-CF & 0.050 & 0.29 & 0.28  &  93.9  \\
TMLE & -0.42 & 0.21 & 0.18 &  37.6  &&IFB & -0.21 & 0.23 & 0.25 &  88.0  \\
TMLE-CF & 0.013 & 0.23 & 0.23  &  95.1 &&WAQ & -0.19 & 0.25 & 0.25 & 88.5 \\
IFB & 0.016 & 0.22 & 0.23 &  96.4 &&&&&&&\\
WAQ & 0.018 & 0.22 & 0.23 & 95.8 
\end{tabular}
\caption{Simulation results for (left) randomized dichotomous exposure and (right) setting 2 of experiment 1: sample size $n=500$, quantile $\tau=0.5$, Monte Carlo bias, Monte Carlo standard deviation (SD), average of the influence function based standard errors (SE) and coverage of 95\% Wald confidence intervals (Cov).}
\label{tab: exp3}
\end{table}

However, when applying the two estimators from \cite{athey2021semiparametric} in the heteroscedasticity setting from experiment 1, we observe a much worse performance (see Table \ref{tab: exp3} (right)). Here the location shift assumption is violated, leading to an inflated bias. As expected, our proposal outperforms their estimators. The results of the TMLE-CF estimator are taken from Table \ref{table1} and included in Table \ref{tab: exp3} (right) for comparison.

\subsection{Experiment 4: Variable selection in higher dimensions}
In this simulation experiment, we considered a covariate vector of dimension 50, based on the data-generating mechanism in \cite{belloni2013}. In particular, we generated a 50-dimensional covariate vector $L\sim\mathcal{N}(0,\Sigma)$, where $\Sigma$ was an autoregressive correlation matrix with correlation parameter 0.5. A continuous exposure was generated according to a normal distribution with mean $\sum_{k=1}^{10}L_k/k$ and unit residual variance. The outcome $Y$ was generated from $Y = A + \sum_{k=1}^{5}L_k/k + \sum_{k=11}^{15}L_k/(k-10) + \epsilon$, with $\epsilon\sim\mathcal{N}(0,2^2)$ independent from $A$ and $L$. Here, we also compare sample size $n=250$ with $n=500$.

In this experiment, we compare the performance of a quantile regression with AIC-based backward stepwise variable selection (QRvs) with the TMLE estimator defined in (\ref{TMLEvs}) with 5-fold cross-fitting (TMLEvs-CF). The results are shown in Table \ref{tab: exp4}. 

Appendix \ref{appD} contains the entire table with additional results for the estimators used in experiment 1, the DML estimator defined in (\ref{DMLvs}) where quantile random forests are replaced by quantile regression with AIC-based backward stepwise variable selection and the TMLE estimator defined in (\ref{TMLEvs}) without cross-fitting.

\begin{table}[h]
\centering
\addtolength{\leftskip}{-2cm}
\addtolength{\rightskip}{-2cm}
\begin{tabular}{cccccccccccc}
& & & \multicolumn{4}{c}{$\mathbf{n=250}$} &  & \multicolumn{4}{c}{$\mathbf{n=500}$} \\ 
\multirow{-2}{*}{\textbf{Quantile}} & \multirow{-2}{*}{\textbf{estimator}} &  & \textbf{bias} & \textbf{SD} & \textbf{SE} & \textbf{Cov} &  & \multicolumn{1}{c}{\textbf{bias}} & \multicolumn{1}{c}{\textbf{SD}} & \multicolumn{1}{c}{\textbf{SE}} & \multicolumn{1}{c}{\textbf{Cov}} \\
\multirow{3}{*}{$\tau=0.5$} & Oracle && 0.66 & 15 & 16 & 96.0 &  &  0.25 & 10 & 11 & 95.7     \\
 & QRvs && 2.4 & 17 & 16 & 91.9 &  &  0.36 & 12 & 11 & 92.8      \\
 & TMLEvs-CF && 1.3 & 20 & 19 & 94.3 &  &  -0.32 & 13 & 12 &   94.8     \\
\multirow{3}{*}{$\tau=0.75$} & Oracle && -0.016 & 16 & 18 & 96.7 &  &  -0.54 & 12 & 12 &   95.0   \\
 & QRvs && -1.4 & 19 & 17 & 91.3 &  &  0.0095 & 13 & 12 &  93.2     \\
 & TMLEvs-CF && 0.52 & 21 & 21 & 94.0 &  &  -0.28 & 14 & 13 &   93.5     \\
\multirow{3}{*}{$\tau=0.9$} & Oracle && -0.88 & 21 & 22 & 95.7 &  &  0.39 & 14 & 15 &  95.8  \\
 & QRvs && 0.065 & 23 & 20 & 89.7 &  &  0.094 & 16 & 14 & 90.8      \\
 & TMLEvs-CF && 0.21 & 24 & 24 & 94.4 &  &  -0.18 & 16 & 16 &   93.7  
\end{tabular}
\caption{Simulation results for variable selection: sample size $n$, quantile $\tau$, Monte Carlo bias, Monte Carlo standard deviation (SD), average of the influence function based standard errors (SE) and coverage of 95\% Wald confidence intervals (Cov). All values have been multiplied by $10^2$.}
\label{tab: exp4}
\end{table}    

We observe that the standard quantile regression with variable selection (QRvs) performs well in terms of bias. However, there is a slight underestimation of the standard errors, leading to a coverage slightly below 95\%. On the contrary, the TMLE estimator achieves a coverage of approximately 94\%. As expected, its standard errors are more accurate because it takes the excess uncertainty of the variable selection procedure into account.

\section{Data analysis}\label{sec:data}
To illustrate our proposal, we reanalyze an observational study to estimate annual health care and lost productivity costs associated with excess weight among the adult population in Belgium, using national health data.
\cite{gorasso2022health} estimated the average difference in health care costs between individuals with underweight (BMI $< 18.5$) and normal weight ($18.5\leq$ BMI $<25$) versus overweight ($25\leq$ BMI $<30$) and obesity (BMI $\geq 30$) using a linear main effect model. They used a double-selection approach to select potential confounders \citep{doubleselection} and estimated a population level effect using g-computation \citep{robinsGcomp}. Candidate variables included age groups, gender, household educational level, household level of income, and some behavioral risk factors such as smoking, alcohol misuse, unhealthy eating behavior, and physical inactivity.

Since health care costs are right-skewed, the average effect of overweight and obesity could be highly influenced by extreme costs. Therefore, analyzing quantiles could add more insight. We will use the proposed TMLE estimator (\ref{TMLEvs}) with 5-fold cross-fitting to assess the conditional association between health care cost and body weight (dichotomized as underweight and normal weight vs.\ overweight and obesity), considering the same set of candidate baseline variables as \cite{gorasso2022health}. This estimator makes use of parametric quantile regression with AIC-based backward stepwise variable selection. We will compare the results with a traditional parametric quantile model that uses the same variable selection procedure. All analyses include sampling weights to adjust for the multi-stage sampling design \citep{bhis}.

\begin{figure}[h]
    \centering
    \includegraphics[width=1\textwidth]{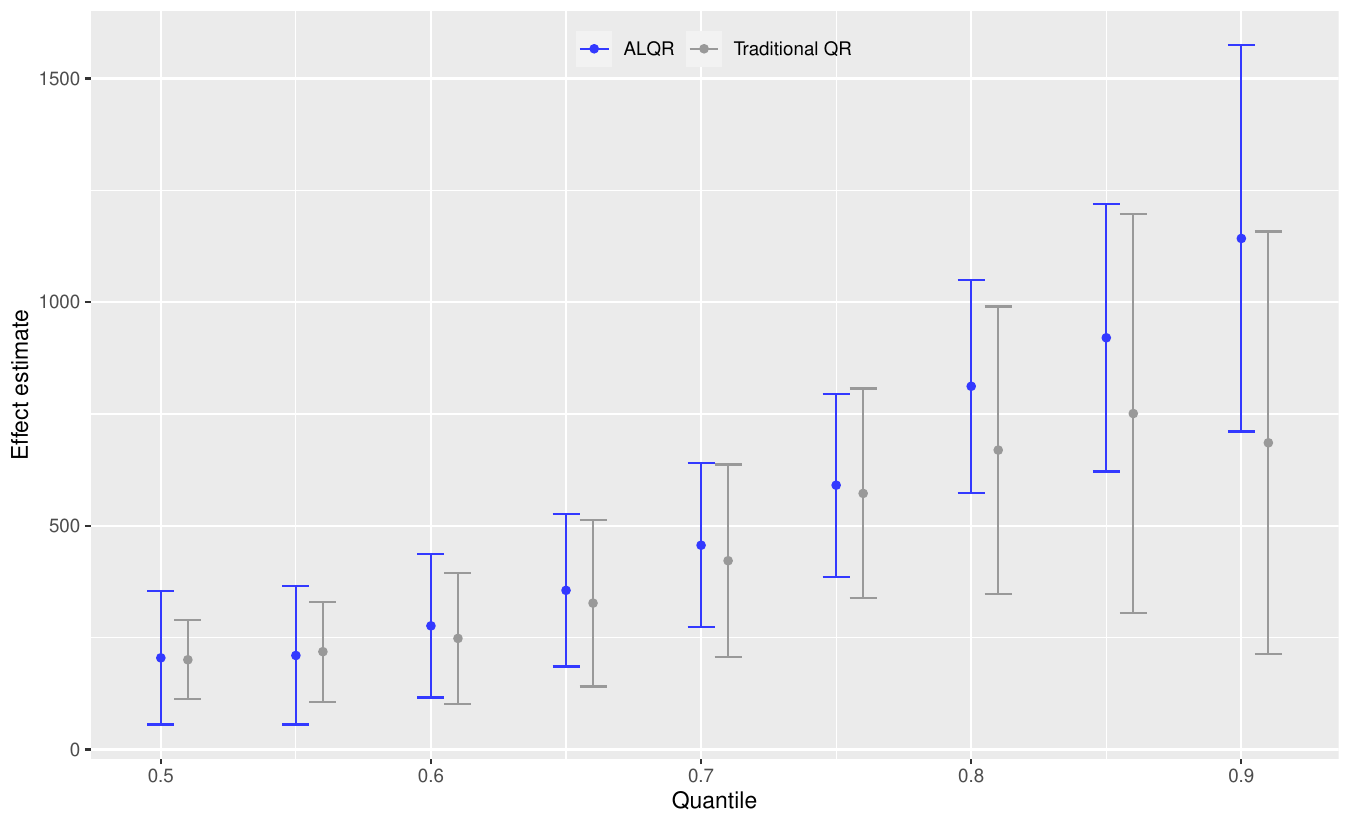}
    \caption{Effect estimate of overweight and obesity on annual health care costs per capita, with 95\% confidence intervals for the proposed method (ALQR) and a traditional quantile regression approach (Traditional QR).}
    \label{fig: sciensano}
\end{figure}

With access to complete data of 3925 individuals, the traditional parametric quantile regression with stepwise variable selection delivers an estimate of the median incremental health care cost per capita of a high BMI of \euro 187.70 (SE: 53.95) and a 90th percentile equal to \euro 521.13 (SE: 259.05). The proposed TMLE approach delivers an estimated median incremental health care cost per capita of excess body weight of \euro 205.09 (SE: 75.94) and an estimate of \euro 1142.42 (SE: 220.15) for the 90th percentile. The fact that our proposal yields different results may be caused by misspecification of the parametric quantile regression model, which did not include interactions involving BMI class. For instance, there is a significant interaction between BMI class and education level ($p=0.0023$). Results for more quantiles are shown in Fig.\ \ref{fig: sciensano}. We observe that the difference in annual health care cost for high versus low body weight is much larger when looking at higher quantiles of the conditional outcome distribution. A sensitivity analysis, which examines additional variability for extreme quantiles related to randomness in the cross-fitting, is included in Appendix \ref{appE}.

\section{Discussion}
In this article, inspired by \cite{AssumptionLean}, we have focused on providing a scalar quantile-based summary of the conditional association between an outcome and an exposure, given covariates. Our results are primarily of interest in settings where averages are less informative and interest rather lies in median (or other quantile) effects (e.g.\ for heavy-tailed or skewed data). The main motivation to focus on conditional quantiles is because marginal quantiles of counterfactual outcomes often demand strong extrapolations by considering a hypothetical scenario where everyone got (or did not get) treatment. This is problematic when some covariate groups are highly unlikely to get treatment. We accommodate this by letting those covariate groups contribute little to our estimand via a very small weight, as is also common in standard regression methods. These groups contain very little information on the exposure effect and there may often be less interest in the effect for those individuals. However, we faced challenges because of the conditional density $f\{Q_\tau(Y\mid A,L)\mid A,L\}$ appearing as a nuisance parameter in the estimator. Estimation of this density is generally a difficult task. Like \cite{nnplqr}, we assume homoscedastic residuals to estimate it.
The `FKSUM' estimator for instance assumes this conditional density to be equal for all subjects. One way to sidestep this assumption is to estimate the (marginal) distribution of $\{Y-E(Y\mid A,L)\}\Var(Y\mid A,L)^{-1/2}$ and then transform an estimate of this distribution to estimate $f\{Q_\tau(Y\mid A,L)\mid A,L\}$. However, this approach resulted in an extremely unstable behaviour of our estimator for $\Psi_\tau$ (not shown) because of estimates of this density close to zero.
The alternative HAL estimator \citep{HAL2016,vanderLaanHAL2017} is more flexible, and does not suffer from instability, but is computationally very expensive. Since the performance of the TMLE estimator did not significantly improve when using the latter density estimator (even in the absence of a location shift), we opted for the simpler one. A broader exploration of different (conditional) density estimators may lead to further improvements in the proposed estimator.
Unfortunately, the existing literature on inverse density weighting is not directly applicable \citep{huling2023}, because we need to evaluate the density in a quantile instead of in the observed outcomes. In future work, we will examine if Auto-DML \citep{chernozhukov2020riesz,chernozhukov2022autodml} could offer a solution to directly estimate the inverse density rather than first estimating the density and then taking the reciprocal. 

From the simulation results for more extreme quantiles, it was often observed that the average estimated standard error is less than the Monte Carlo standard deviation of the estimator. This could be the result of ignoring the randomness in cross-fitting, for which rank-transformed subsampling \citep{guo2023rank} could offer improvements.

Our proposed estimators allow for estimating the nuisance parameters through data-adaptive techniques. To estimate $Q_\tau(Y\mid A,L)$, we utilized either quantile regression forests or parametric quantile regression with variable selection. A potential improvement could be achieved by using an ensemble learner, e.g.\ by constructing an ensemble through the R package `SuperLearner', including the aforementioned methods in its library, potentially supplemented with other estimation methods (e.g.\ Lasso penalized quantile regression). In this way, the user no longer needs to manually choose which algorithm to use. However, to our knowledge such ensemble learners for quantile estimation are currently unavailable.

Finally, it could be useful to generalize the estimand (\ref{estimand3}) so that it includes a link function and replace model (\ref{model}) by $g\{Q_\tau(Y\mid A,L)\} = \beta_\tau A + \omega_\tau (L)$, with link function $g$. More details about this proposal can be found in Appendix \ref{appF}.

\bibliographystyle{unsrtnat}
\bibliography{references}

\begin{thebibliography}{62}
\providecommand{\natexlab}[1]{#1}
\providecommand{\url}[1]{\texttt{#1}}
\expandafter\ifx\csname urlstyle\endcsname\relax
  \providecommand{\doi}[1]{doi: #1}\else
  \providecommand{\doi}{doi: \begingroup \urlstyle{rm}\Url}\fi

\bibitem[Koenker and Bassett(1978)]{koenkerbassett}
Roger Koenker and Gilbert Bassett.
\newblock Regression {Q}uantiles.
\newblock \emph{Econometrica}, 46\penalty0 (1):\penalty0 33--50, 1978.
\newblock ISSN 00129682, 14680262.
\newblock URL \url{http://www.jstor.org/stable/1913643}.

\bibitem[Koenker(2005)]{koenker2005quantile}
Roger Koenker.
\newblock \emph{Quantile {R}egression}, volume~38.
\newblock Cambridge University Press, 2005.

\bibitem[Koenker(2017)]{koenker2017}
Roger Koenker.
\newblock Quantile regression: 40 years on.
\newblock \emph{Annual Review of Economics}, 9\penalty0 (1):\penalty0 155--176, 2017.
\newblock \doi{10.1146/annurev-economics-063016-103651}.
\newblock URL \url{https://doi.org/10.1146/annurev-economics-063016-103651}.

\bibitem[Koenker et~al.(2017)Koenker, Chernozhukov, He, and Peng]{koenker2017handbook}
Roger Koenker, Victor Chernozhukov, Xuming He, and Limin Peng.
\newblock \emph{Handbook of quantile regression}.
\newblock Chapman \& Hall/CRC, New York, 1st edition, 2017.

\bibitem[Koenker and Hallock(2001)]{koenkerhallock2001}
Roger Koenker and Kevin~F. Hallock.
\newblock Quantile {R}egression.
\newblock \emph{Journal of Economic Perspectives}, 15\penalty0 (4):\penalty0 143--156, 2001.
\newblock \doi{10.1257/jep.15.4.143}.
\newblock URL \url{https://www.aeaweb.org/articles?id=10.1257/jep.15.4.143}.

\bibitem[Yu et~al.(2003)Yu, Lu, and Stander]{YuLuStander}
Keming Yu, Zudi Lu, and Julian Stander.
\newblock Quantile regression: applications and current research areas.
\newblock \emph{Journal of the Royal Statistical Society: Series D (The Statistician)}, 52\penalty0 (3):\penalty0 331--350, 2003.
\newblock \doi{https://doi.org/10.1111/1467-9884.00363}.
\newblock URL \url{https://rss.onlinelibrary.wiley.com/doi/abs/10.1111/1467-9884.00363}.

\bibitem[Das et~al.(2019)Das, Krzywinski, and Altman]{das2019quantile}
Kiranmoy Das, Martin Krzywinski, and Naomi Altman.
\newblock Quantile regression.
\newblock \emph{Nature Methods}, 16\penalty0 (6):\penalty0 451--452, 2019.

\bibitem[Hern{\'a}n and Robins(2020)]{hernan2020causal}
Miguel~A Hern{\'a}n and James~M Robins.
\newblock \emph{Causal inference: {W}hat {I}f}.
\newblock Boca Raton: Chapman \& Hall/CRC, 2020.

\bibitem[Doksum(1974)]{doksum1974}
Kjell Doksum.
\newblock Empirical {P}robability {P}lots and {S}tatistical {I}nference for {N}onlinear {M}odels in the {T}wo-{S}ample {C}ase.
\newblock \emph{The Annals of Statistics}, 2\penalty0 (2):\penalty0 267--277, 1974.
\newblock ISSN 00905364.
\newblock URL \url{http://www.jstor.org/stable/2958036}.

\bibitem[Lehmann and D'Abrera(1975)]{lehmann1975}
E.~L. Lehmann and H.~J. D'Abrera.
\newblock \emph{Nonparametrics: {S}tatistical methods based on ranks}.
\newblock Holden-Day, 1975.

\bibitem[Vansteelandt and Dukes(2022)]{AssumptionLean}
Stijn Vansteelandt and Oliver Dukes.
\newblock {Assumption-lean {I}nference for {G}eneralised {L}inear {M}odel {P}arameters (with discussion)}.
\newblock \emph{Journal of the Royal Statistical Society Series B: Statistical Methodology}, 84\penalty0 (3):\penalty0 657--685, 07 2022.
\newblock ISSN 1369-7412.
\newblock \doi{10.1111/rssb.12504}.
\newblock URL \url{https://doi.org/10.1111/rssb.12504}.

\bibitem[Breiman(2001)]{breiman2001}
Leo Breiman.
\newblock {Statistical Modeling: The Two Cultures (with comments and a rejoinder by the author)}.
\newblock \emph{Statistical Science}, 16\penalty0 (3):\penalty0 199--231, 2001.
\newblock \doi{10.1214/ss/1009213726}.
\newblock URL \url{https://doi.org/10.1214/ss/1009213726}.

\bibitem[Leeb and Pötscher(2005)]{leebpotscher2005}
Hannes Leeb and Benedikt~M. Pötscher.
\newblock Model selection and inference: Facts and fiction.
\newblock \emph{Econometric Theory}, 21\penalty0 (1):\penalty0 21--59, 2005.
\newblock \doi{10.1017/S0266466605050036}.

\bibitem[Melly(2006)]{melly2006}
Blaise Melly.
\newblock Estimation of counterfactual distributions using quantile regression.
\newblock \emph{University of St.Gallen}, 68, 01 2006.

\bibitem[Firpo(2007)]{firpo2007}
Sergio Firpo.
\newblock Efficient {S}emiparametric {E}stimation of {Q}uantile {T}reatment {E}ffects.
\newblock \emph{Econometrica}, 75\penalty0 (1):\penalty0 259--276, 2007.
\newblock \doi{https://doi.org/10.1111/j.1468-0262.2007.00738.x}.
\newblock URL \url{https://onlinelibrary.wiley.com/doi/abs/10.1111/j.1468-0262.2007.00738.x}.

\bibitem[Chernozhukov et~al.(2013)Chernozhukov, Fernández-Val, and Melly]{chernozhukov2013}
Victor Chernozhukov, Iván Fernández-Val, and Blaise Melly.
\newblock Inference on {C}ounterfactual {D}istributions.
\newblock \emph{Econometrica}, 81\penalty0 (6):\penalty0 2205--2268, 2013.
\newblock \doi{https://doi.org/10.3982/ECTA10582}.
\newblock URL \url{https://onlinelibrary.wiley.com/doi/abs/10.3982/ECTA10582}.

\bibitem[Donald and Hsu(2014)]{DONALD2014383}
Stephen~G. Donald and Yu-Chin Hsu.
\newblock Estimation and inference for distribution functions and quantile functions in treatment effect models.
\newblock \emph{Journal of Econometrics}, 178:\penalty0 383--397, 2014.
\newblock ISSN 0304-4076.
\newblock \doi{https://doi.org/10.1016/j.jeconom.2013.03.010}.
\newblock URL \url{https://www.sciencedirect.com/science/article/pii/S0304407613001826}.

\bibitem[Caracciolo and Furno(2017)]{caracciolo2017}
Francesco Caracciolo and Marilena Furno.
\newblock Quantile treatment effect and double robust estimators: an appraisal on the italian labor market.
\newblock \emph{Journal of Economic Studies}, 4\penalty0 (2):\penalty0 585--604, 2017.

\bibitem[Athey et~al.(2023)Athey, Bickel, Chen, Imbens, and Pollmann]{athey2021semiparametric}
Susan Athey, Peter~J Bickel, Aiyou Chen, Guido~W Imbens, and Michael Pollmann.
\newblock {Semi-parametric estimation of treatment effects in randomised experiments}.
\newblock \emph{Journal of the Royal Statistical Society Series B: Statistical Methodology}, 07 2023.
\newblock ISSN 1369-7412.
\newblock \doi{10.1093/jrsssb/qkad072}.
\newblock URL \url{https://doi.org/10.1093/jrsssb/qkad072}.

\bibitem[Ai et~al.(2022)Ai, Linton, and Zhang]{Ai2022}
Chunrong Ai, Oliver Linton, and Zheng Zhang.
\newblock Estimation and inference for the counterfactual distribution and quantile functions in continuous treatment models.
\newblock \emph{Journal of Econometrics}, 228\penalty0 (1):\penalty0 39--61, 2022.
\newblock ISSN 0304-4076.
\newblock \doi{https://doi.org/10.1016/j.jeconom.2020.12.009}.
\newblock URL \url{https://www.sciencedirect.com/science/article/pii/S0304407621000543}.
\newblock Annals Issue: In Honor of Ron Gallant.

\bibitem[Díaz(2017)]{DIAZ201739}
Iván Díaz.
\newblock Efficient estimation of quantiles in missing data models.
\newblock \emph{Journal of Statistical Planning and Inference}, 190:\penalty0 39--51, 2017.
\newblock ISSN 0378-3758.
\newblock \doi{https://doi.org/10.1016/j.jspi.2017.05.001}.
\newblock URL \url{https://www.sciencedirect.com/science/article/pii/S0378375817300782}.

\bibitem[Belloni et~al.(2017)Belloni, Chernozhukov, Fernandez-Val, and Hansen]{belloni2017}
Alexandre Belloni, Victor Chernozhukov, Ivan Fernandez-Val, and Christian Hansen.
\newblock Program evaluation and causal inference with high-dimensional data.
\newblock \emph{Econometrica}, 85\penalty0 (1):\penalty0 233--298, 2017.

\bibitem[Kallus et~al.(2019)Kallus, Mao, and Uehara]{kallus2019localized}
Nathan Kallus, Xiaojie Mao, and Masatoshi Uehara.
\newblock Localized debiased machine learning: Efficient inference on quantile treatment effects and beyond.
\newblock \emph{arXiv preprint arXiv:1912.12945}, 2019.

\bibitem[Chernozhukov et~al.(2018)Chernozhukov, Chetverikov, Demirer, Duflo, Hansen, Newey, and Robins]{chernozhukov2018}
Victor Chernozhukov, Denis Chetverikov, Mert Demirer, Esther Duflo, Christian Hansen, Whitney Newey, and James Robins.
\newblock {Double/debiased machine learning for treatment and structural parameters}.
\newblock \emph{The Econometrics Journal}, 21\penalty0 (1):\penalty0 C1--C68, 01 2018.
\newblock ISSN 1368-4221.
\newblock \doi{10.1111/ectj.12097}.
\newblock URL \url{https://doi.org/10.1111/ectj.12097}.

\bibitem[Lee(2003)]{lee_2003}
Sokbae Lee.
\newblock Efficient semiparametric estimation of a partially linear quantile regression model.
\newblock \emph{Econometric Theory}, 19\penalty0 (1):\penalty0 1–31, 2003.
\newblock \doi{10.1017/S0266466603191013}.

\bibitem[Sun(2005)]{sun2005}
Yiguo Sun.
\newblock Semiparametric efficient estimation of partially linear quantile regression models.
\newblock \emph{Annals of Economics and Finance}, 6\penalty0 (1):\penalty0 105, 2005.

\bibitem[Wu et~al.(2010)Wu, Yu, and Yu]{WU20101607}
Tracy~Z. Wu, Keming Yu, and Yan Yu.
\newblock Single-index quantile regression.
\newblock \emph{Journal of Multivariate Analysis}, 101\penalty0 (7):\penalty0 1607--1621, 2010.
\newblock ISSN 0047-259X.
\newblock \doi{https://doi.org/10.1016/j.jmva.2010.02.003}.
\newblock URL \url{https://www.sciencedirect.com/science/article/pii/S0047259X10000333}.

\bibitem[Wu and Yu(2014)]{WU2014170}
Chaojiang Wu and Yan Yu.
\newblock Partially linear modeling of conditional quantiles using penalized splines.
\newblock \emph{Computational Statistics \& Data Analysis}, 77:\penalty0 170--187, 2014.
\newblock ISSN 0167-9473.
\newblock \doi{https://doi.org/10.1016/j.csda.2014.02.020}.
\newblock URL \url{https://www.sciencedirect.com/science/article/pii/S0167947314000607}.

\bibitem[Lv et~al.(2015)Lv, Zhang, Zhao, and Liu]{lv2015quantile}
Yazhao Lv, Riquan Zhang, Weihua Zhao, and Jicai Liu.
\newblock Quantile regression and variable selection of partial linear single-index model.
\newblock \emph{Annals of the Institute of Statistical Mathematics}, 67\penalty0 (2):\penalty0 375--409, 2015.

\bibitem[Sherwood and Wang(2016)]{SherwoodWang2016}
Ben Sherwood and Lan Wang.
\newblock {Partially linear additive quantile regression in ultra-high dimension}.
\newblock \emph{The Annals of Statistics}, 44\penalty0 (1):\penalty0 288--317, 2016.
\newblock \doi{10.1214/15-AOS1367}.
\newblock URL \url{https://doi.org/10.1214/15-AOS1367}.

\bibitem[Zhong and Wang(2023)]{nnplqr}
Qixian Zhong and Jane-Ling Wang.
\newblock Neural networks for partially linear quantile regression.
\newblock \emph{Journal of Business \& Economic Statistics}, 0\penalty0 (0):\penalty0 1--12, 2023.
\newblock \doi{10.1080/07350015.2023.2208183}.
\newblock URL \url{https://doi.org/10.1080/07350015.2023.2208183}.

\bibitem[He and Shi(1996)]{he1996}
Xuming He and Peide Shi.
\newblock Bivariate tensor-product b-splines in a partly linear model.
\newblock \emph{Journal of Multivariate Analysis}, 58\penalty0 (2):\penalty0 162--181, 1996.

\bibitem[Vansteelandt et~al.(2022)Vansteelandt, Dukes, Lancker, and Martinussen]{alcox}
Stijn Vansteelandt, Oliver Dukes, Kelly~Van Lancker, and Torben Martinussen.
\newblock Assumption-{L}ean {C}ox {R}egression.
\newblock \emph{Journal of the American Statistical Association}, pages 1--10, 2022.
\newblock \doi{10.1080/01621459.2022.2126362}.
\newblock URL \url{https://doi.org/10.1080/01621459.2022.2126362}.

\bibitem[van~der Laan and Rose(2011)]{vanderlaanrose2011}
Mark~J van~der Laan and Sherri Rose.
\newblock \emph{Targeted learning: causal inference for observational and experimental data}, volume~4.
\newblock Springer Series in Statistics, New York, 2011.

\bibitem[Pfanzagl(1990)]{Pfanzagl1990}
Johann Pfanzagl.
\newblock \emph{Estimation in semiparametric models}, pages 17--22.
\newblock Springer US, New York, NY, 1990.
\newblock ISBN 978-1-4612-3396-1.
\newblock \doi{10.1007/978-1-4612-3396-1\_5}.
\newblock URL \url{https://doi.org/10.1007/978-1-4612-3396-1\_5}.

\bibitem[Bickel et~al.(1993)Bickel, Klaassen, Bickel, Ritov, Klaassen, Wellner, and Ritov]{bickel1993efficient}
Peter~J Bickel, Chris~AJ Klaassen, Peter~J Bickel, Ya’acov Ritov, J~Klaassen, Jon~A Wellner, and YA'Acov Ritov.
\newblock \emph{Efficient and adaptive estimation for semiparametric models}, volume~4.
\newblock Baltimore: John Hopkins University Press, 1993.

\bibitem[Zheng and van~der Laan(2011)]{Zheng2011}
Wenjing Zheng and Mark~J. van~der Laan.
\newblock \emph{Cross-Validated Targeted Minimum-Loss-Based Estimation}, pages 459--474.
\newblock Springer New York, New York, NY, 2011.
\newblock ISBN 978-1-4419-9782-1.
\newblock \doi{10.1007/978-1-4419-9782-1\_27}.
\newblock URL \url{https://doi.org/10.1007/978-1-4419-9782-1\_27}.

\bibitem[van~der Vaart(1998)]{vdv2000asymptotic}
A.~W. van~der Vaart.
\newblock \emph{Functional Delta Method}, page 291–303.
\newblock Cambridge Series in Statistical and Probabilistic Mathematics. Cambridge University Press, 1998.

\bibitem[van~der Laan and Rubin(2006)]{vanderLaanRubin+2006}
Mark~J. van~der Laan and Daniel Rubin.
\newblock Targeted {M}aximum {L}ikelihood {L}earning.
\newblock \emph{The International Journal of Biostatistics}, 2\penalty0 (1):\penalty0 Article 11, 2006.
\newblock \doi{10.2202/1557-4679.1043}.
\newblock URL \url{https://doi.org/10.2202/1557-4679.1043}.

\bibitem[Yadlowsky(2022)]{yadlowsky2022}
Steve Yadlowsky.
\newblock Explaining practical differences between treatment effect estimators with high dimensional asymptotics.
\newblock \emph{arXiv preprint arXiv:2203.12538}, 2022.

\bibitem[Rosenbaum and Rubin(1983)]{rosenbaum1983}
Paul~R. Rosenbaum and Donald~B. Rubin.
\newblock {The central role of the propensity score in observational studies for causal effects}.
\newblock \emph{Biometrika}, 70\penalty0 (1):\penalty0 41--55, 04 1983.
\newblock ISSN 0006-3444.
\newblock \doi{10.1093/biomet/70.1.41}.
\newblock URL \url{https://doi.org/10.1093/biomet/70.1.41}.

\bibitem[Rosenbaum(1989)]{rosenbaum1989}
Paul~R. Rosenbaum.
\newblock Optimal matching for observational studies.
\newblock \emph{Journal of the American Statistical Association}, 84\penalty0 (408):\penalty0 1024--1032, 1989.
\newblock \doi{10.1080/01621459.1989.10478868}.
\newblock URL \url{https://www.tandfonline.com/doi/abs/10.1080/01621459.1989.10478868}.

\bibitem[Sun et~al.(2021)Sun, Moodie, and Neslehová]{Sun2021}
Shuo Sun, Erica E.~M. Moodie, and Johanna~G. Neslehová.
\newblock Causal inference for quantile treatment effects.
\newblock \emph{Environmetrics}, 32\penalty0 (4):\penalty0 e2668, 2021.
\newblock \doi{https://doi.org/10.1002/env.2668}.
\newblock URL \url{https://onlinelibrary.wiley.com/doi/abs/10.1002/env.2668}.

\bibitem[Athey et~al.(2019)Athey, Tibshirani, and Wager]{athey2019generalized}
Susan Athey, Julie Tibshirani, and Stefan Wager.
\newblock {Generalized random forests}.
\newblock \emph{The Annals of Statistics}, 47\penalty0 (2):\penalty0 1148--1178, 2019.
\newblock \doi{10.1214/18-AOS1709}.
\newblock URL \url{https://doi.org/10.1214/18-AOS1709}.

\bibitem[van~der Laan et~al.(2007)van~der Laan, Polley, and Hubbard]{vanderLaanSL2007}
Mark~J. van~der Laan, Eric~C Polley, and Alan~E. Hubbard.
\newblock Super {L}earner.
\newblock \emph{Statistical Applications in Genetics and Molecular Biology}, 6\penalty0 (1):\penalty0 Article 25, 2007.
\newblock \doi{doi:10.2202/1544-6115.1309}.
\newblock URL \url{https://doi.org/10.2202/1544-6115.1309}.

\bibitem[Hofmeyr(2021)]{fksum1}
David~P. Hofmeyr.
\newblock Fast {E}xact {E}valuation of {U}nivariate {K}ernel {S}ums.
\newblock \emph{IEEE Transactions on Pattern Analysis and Machine Intelligence}, 43\penalty0 (2):\penalty0 447--458, 2021.
\newblock \doi{10.1109/TPAMI.2019.2930501}.

\bibitem[Hofmeyr(2022)]{fksum2}
David~P. Hofmeyr.
\newblock Fast {K}ernel {S}moothing in {R} with {A}pplications to {P}rojection {P}ursuit.
\newblock \emph{Journal of Statistical Software}, 101\penalty0 (3):\penalty0 1--33, 2022.
\newblock \doi{10.18637/jss.v101.i03}.

\bibitem[Hejazi et~al.(2022{\natexlab{a}})Hejazi, Benkeser, D{\'\i}az, and {van der Laan}]{hejazi2022efficient}
Nima~S Hejazi, David Benkeser, Iv{\'a}n D{\'\i}az, and Mark~J {van der Laan}.
\newblock Efficient estimation of modified treatment policy effects based on the generalized propensity score.
\newblock 2022{\natexlab{a}}.
\newblock URL \url{https://arxiv.org/abs/2205.05777}.

\bibitem[Hejazi et~al.(2022{\natexlab{b}})Hejazi, Benkeser, and {van der Laan}]{hejazi2022haldensify-rpkg}
Nima~S Hejazi, David~C Benkeser, and Mark~J {van der Laan}.
\newblock {haldensify}: {H}ighly adaptive lasso conditional density estimation.
\newblock \url{https://github.com/nhejazi/haldensify}, 2022{\natexlab{b}}.
\newblock URL \url{https://doi.org/10.5281/zenodo.3698329}.
\newblock {R} package version 0.2.5.

\bibitem[Hejazi et~al.(2022{\natexlab{c}})Hejazi, {van der Laan}, and Benkeser]{hejazi2022haldensify-joss}
Nima~S Hejazi, Mark~J {van der Laan}, and David~C Benkeser.
\newblock {haldensify}: {H}ighly adaptive lasso conditional density estimation in {R}.
\newblock \emph{Journal of Open Source Software}, 2022{\natexlab{c}}.
\newblock \doi{10.21105/joss.04522}.
\newblock URL \url{https://doi.org/10.21105/joss.04522}.

\bibitem[Benkeser and Van Der~Laan(2016)]{HAL2016}
David Benkeser and Mark Van Der~Laan.
\newblock The {H}ighly {A}daptive {L}asso {E}stimator.
\newblock In \emph{2016 IEEE International Conference on Data Science and Advanced Analytics (DSAA)}, pages 689--696, 2016.
\newblock \doi{10.1109/DSAA.2016.93}.

\bibitem[van~der Laan(2017)]{vanderLaanHAL2017}
Mark~J van~der Laan.
\newblock A {G}enerally {E}fficient {T}argeted {M}inimum {L}oss {B}ased {E}stimator based on the {H}ighly {A}daptive {L}asso.
\newblock \emph{The International Journal of Biostatistics}, 13\penalty0 (2):\penalty0 20150097, 2017.
\newblock \doi{doi:10.1515/ijb-2015-0097}.
\newblock URL \url{https://doi.org/10.1515/ijb-2015-0097}.

\bibitem[Belloni et~al.(2013{\natexlab{a}})Belloni, Chernozhukov, and Wei]{belloni2013}
Alexandre Belloni, Victor Chernozhukov, and Ying Wei.
\newblock Honest confidence regions for a regression parameter in logistic regression with a large number of controls.
\newblock cemmap working paper CWP67/13, Centre for Microdata Methods and Practice (cemmap), London, 2013{\natexlab{a}}.
\newblock URL \url{http://hdl.handle.net/10419/97416}.

\bibitem[Gorasso et~al.(2022)Gorasso, Moyersoen, Van~der Heyden, De~Ridder, Vandevijvere, Vansteelandt, De~Smedt, and Devleesschauwer]{gorasso2022health}
Vanessa Gorasso, Isabelle Moyersoen, Johan Van~der Heyden, Karin De~Ridder, Stefanie Vandevijvere, Stijn Vansteelandt, Delphine De~Smedt, and Brecht Devleesschauwer.
\newblock Health care costs and lost productivity costs related to excess weight in {B}elgium.
\newblock \emph{BMC Public Health}, 22\penalty0 (1):\penalty0 1--11, 2022.

\bibitem[Belloni et~al.(2013{\natexlab{b}})Belloni, Chernozhukov, and Hansen]{doubleselection}
Alexandre Belloni, Victor Chernozhukov, and Christian Hansen.
\newblock {Inference on Treatment Effects after Selection among High-Dimensional Controls†}.
\newblock \emph{The Review of Economic Studies}, 81\penalty0 (2):\penalty0 608--650, 11 2013{\natexlab{b}}.
\newblock ISSN 0034-6527.
\newblock \doi{10.1093/restud/rdt044}.
\newblock URL \url{https://doi.org/10.1093/restud/rdt044}.

\bibitem[Robins(1986)]{robinsGcomp}
James Robins.
\newblock A new approach to causal inference in mortality studies with a sustained exposure period—application to control of the healthy worker survivor effect.
\newblock \emph{Mathematical Modelling}, 7\penalty0 (9-12):\penalty0 1393--1512, 1986.

\bibitem[Demarest et~al.(2013)Demarest, Van~der Heyden, Charafeddine, Drieskens, Gisle, and Tafforeau]{bhis}
Stefaan Demarest, Johan Van~der Heyden, Rana Charafeddine, Sabine Drieskens, Lydia Gisle, and Jean Tafforeau.
\newblock Methodological basics and evolution of the {B}elgian health interview survey 1997--2008.
\newblock \emph{Archives of Public Health}, 71:\penalty0 1--10, 2013.

\bibitem[Huling et~al.(2023)Huling, Greifer, and Chen]{huling2023}
Jared~D Huling, Noah Greifer, and Guanhuan Chen.
\newblock Independence weights for causal inference with continuous treatments.
\newblock \emph{Journal of the American Statistical Association}, 0\penalty0 (0):\penalty0 1--14, 2023.
\newblock \doi{10.1080/01621459.2023.2213485}.
\newblock URL \url{https://doi.org/10.1080/01621459.2023.2213485}.

\bibitem[Chernozhukov et~al.(2020)Chernozhukov, Newey, Singh, and Syrgkanis]{chernozhukov2020riesz}
Victor Chernozhukov, Whitney Newey, Rahul Singh, and Vasilis Syrgkanis.
\newblock Adversarial estimation of {R}iesz representers.
\newblock \emph{arXiv preprint arXiv:2101.00009}, 2020.

\bibitem[Chernozhukov et~al.(2022)Chernozhukov, Newey, and Singh]{chernozhukov2022autodml}
Victor Chernozhukov, Whitney~K Newey, and Rahul Singh.
\newblock Automatic debiased machine learning of causal and structural effects.
\newblock \emph{Econometrica}, 90\penalty0 (3):\penalty0 967--1027, 2022.

\bibitem[Guo and Shah(2023)]{guo2023rank}
F~Richard Guo and Rajen~D Shah.
\newblock Rank-transformed subsampling: inference for multiple data splitting and exchangeable p-values.
\newblock \emph{arXiv preprint arXiv:2301.02739}, 2023.

\bibitem[Hines et~al.(2022)Hines, Dukes, Diaz-Ordaz, and Vansteelandt]{hines2022}
Oliver Hines, Oliver Dukes, Karla Diaz-Ordaz, and Stijn Vansteelandt.
\newblock Demystifying {S}tatistical {L}earning {B}ased on {E}fficient {I}nfluence {F}unctions.
\newblock \emph{The American Statistician}, 76\penalty0 (3):\penalty0 292--304, 2022.
\newblock \doi{10.1080/00031305.2021.2021984}.
\newblock URL \url{https://doi.org/10.1080/00031305.2021.2021984}.

\end{thebibliography}

\newpage
\appendix
\noindent\huge{\textbf{Appendices}}
\normalsize
\renewcommand{\thesection}{\Alph{section}}
\section{Appendix A} \label{appA}
\subsection{Proof that (\ref{estimand2}) reduces to $\beta_\tau$ if model (\ref{model}) holds.}
Model (\ref{model}) states that $Q_\tau(Y|A,L) = \beta_\tau A + \omega_\tau(L)$. Substituting $Q_\tau(Y|A,L)$ and $Q_\tau(Y|A^*,L)$ in (\ref{estimand2}) leads to 
\begin{align*}
    \Psi_\tau &= \frac{E[(A-A^*)\{Q_\tau(Y|A,L)-Q_\tau(Y|A^*,L)\}]}{E\{(A-A^*)^2\}}\\
    &= \frac{E[(A-A^*)\{\beta_\tau A + \omega_\tau(L)-\beta_\tau A^* + \omega_\tau(L)\}]}{E\{(A-A^*)^2\}}\\
    &= \frac{E[(A-A^*)\{\beta_\tau (A - A^*)\}]}{E\{(A-A^*)^2\}}\\
    &= \beta_\tau.
\end{align*}

\subsection{Proof that the estimands (\ref{estimand2}) and (\ref{estimand3}) are equivalent}
We start from the numerator of (\ref{estimand2}):
\begin{align*}
    &E[(A-A^*)\{Q_\tau(Y|A,L)-Q_\tau(Y|A^*,L)\}]\\
    &= E[\{A-E(A|L)+E(A|L)-A^*\}\{Q_\tau(Y|A,L)-Q_\tau(Y|A^*,L)\}]\\
    &= E[\{A-E(A|L)\}\{Q_\tau(Y|A,L)-Q_\tau(Y|A^*,L)\}] + E[\{A^*-E(A|L)\}\{Q_\tau(Y|A^*,L)-Q_\tau(Y|A,L)\}]\\
    &= 2E[\{A-E(A|L)\}Q_\tau(Y|A,L)] - E[\{A-E(A|L)\}Q_\tau(Y|A^*,L)] - E[\{A^*-E(A^*|L)\}Q_\tau(Y|A,L)].
\end{align*}
We made use of the fact that $A$ and $A^*$ follow the same distribution conditional on $L$. Moreover, the last two terms are zero because of the (assumed) conditional independence of $A$ and $A^*$ given $L$:
\begin{align*}
    E[\{A-E(A|L)\}Q_\tau(Y|A^*,L)] &= E\big(E[\{A-E(A|L)\}Q_\tau(Y|A^*,L)|L]\big)\\
    &= E\big(E[\{A-E(A|L)\}|L]E\{Q_\tau(Y|A^*,L)|L\}\big)\\
    &= E\big[0E\{Q_\tau(Y|A^*,L)|L\}\big]\\
    &= 0.
\end{align*}
The numerator of (\ref{estimand3}) reduces to
\[
E\left(\{A-E(A|L)\}\left[Q_\tau(Y|A,L)-E\{Q_\tau(Y|A,L)|L\}\right]\right) = E\left[\{A-E(A|L)\}Q_\tau(Y|A,L)\right].
\]
Thus, the numerator of (\ref{estimand2}) is equal to 2 times the numerator of (\ref{estimand3}).
Now we show that also the denominator of (\ref{estimand2}) is equal to 2 times the denominator of (\ref{estimand3}), such that (\ref{estimand2}) and (\ref{estimand3}) are equivalent:
\begin{align*}
    E\{(A-A^*)^2\} &= E[\{A-E(A|L)+E(A|L)-A^*\}^2]\\
    &= E[\{A-E(A|L)\}^2+\{E(A|L)-A^*\}^2-2\{A-E(A|L)\}\{A^*-E(A|L)\}]\\
    &= 2E[\{A-E(A|L)\}^2] - 2E[\{A-E(A|L)\}\{A^*-E(A|L)\}]\\
    &= 2E[\{A-E(A|L)\}^2] - 2E\big(E[\{A-E(A|L)\}\{A^*-E(A|L)\}|L]\big)\\
    &= 2E[\{A-E(A|L)\}^2] - 2E\big(E[\{A-E(A|L)\}|L]E[\{A^*-E(A|L)\}|L]\big)\\
    &= 2E[\{A-E(A|L)\}^2].
\end{align*}

\newpage
\section{Appendix B}\label{appB}
\subsection{Proof of Theorem \ref{theorem1}}
First note that estimand (\ref{estimand3}) can be equivalently written as
\begin{equation}\label{estimand}
\Psi_\tau = \frac{E[\{A-E(A|L)\}Q_\tau(Y|A,L)]}{E[\{A-E(A|L)\}^2]}.
\end{equation}
We will use (\ref{estimand}) to calculate the efficient influence function of (\ref{estimand3}). Furthermore, we follow the strategy described in \cite{hines2022} as well as the notation. Therefore, consider the one-dimensional parametric submodel $\mathcal{P}_t = (1-t)\mathcal{P}+t\widetilde{\mathcal{P}}$, where $\mathcal{P}$ is the true data-generating law and $\widetilde{\mathcal{P}}$ is a point mass data generating distribution, such that the density of $\mathcal{P}_t$ is given by $f_t(y,a,l)=(1-t)f(y,a,l)+t\mathbbm{1}_{(\Tilde{y},\Tilde{a};\Tilde{l})}(y,a,l)$, where $\mathbbm{1}_{(\Tilde{y},\Tilde{a};\Tilde{l})}(y,a,l)$ is a Dirac delta function and $f(y,a,l)$ is the joint density of $Y$, $A$ and $L$. In the same way, we also define all other functions, e.g. $f_t(a|l)=f_t(a,l)/f_t(l)$. To calculate the efficient influence function, we first give some intermediate results.

\begin{lemma}\label{lemma1}
\begin{align*}
    \frac{d}{dt} Q_{\tau,t}(Y|A,L)\Big\lvert_{t=0} &= \frac{\mathbbm{1}_{\Tilde{a},\Tilde{l}}(A,L)}{f\{Q_\tau(Y|A,L),A,L\}}\Big[ \tau-I\{\Tilde{y}\leq Q_\tau(Y|A,L)\}\Big]
\end{align*}
\begin{proof}
Starting from 
\[
\tau = \int_{-\infty}^{Q_\tau(Y|A,L)}f(Y|A,L)\,dY = \int_{-\infty}^{Q_\tau(Y|A,L)}\frac{f(Y,A,L)}{f(A,L)}\,dY
\]
and using the Leibniz integral rule, one can find
\begin{align*}
    0 &= \frac{d}{dt} \left( \int_{-\infty}^{Q_{\tau,t}(Y|A,L)}\frac{f_t(Y,A,L)}{f_t(A,L)}\,dY \right) \Bigg\lvert_{t=0} \\
    &= \frac{f(Q_\tau(Y|A,L),A,L)}{f(A,L)} \frac{d}{dt} Q_{\tau,t}(Y|A,L)\Big\lvert_{t=0} + \int_{-\infty}^{Q_\tau(Y|A,L)}\frac{\mathbbm{1}_{\Tilde{y},\Tilde{a},\Tilde{l}}(Y,A,L)-f(Y,A,L)}{f(A,L)}\,dY\\
    & \qquad\qquad\qquad - \int_{-\infty}^{Q_\tau(Y|A,L)}\frac{[\mathbbm{1}_{\Tilde{a},\Tilde{l}}(A,L)-f(A,L)]f(Y,A,L)}{f(A,L)^2}\,dY\\
    &= \frac{f(Q_\tau(Y|A,L),A,L)}{f(A,L)} \frac{d}{dt} Q_{\tau,t}(Y|A,L)\Big\lvert_{t=0} + \int_{-\infty}^{Q_\tau(Y|A,L)}\frac{\mathbbm{1}_{\Tilde{y},\Tilde{a},\Tilde{l}}(Y,A,L)}{f(A,L)}\,dY\\
    & \qquad\qquad\qquad - \int_{-\infty}^{Q_\tau(Y|A,L)}\frac{\mathbbm{1}_{\Tilde{a},\Tilde{l}}(A,L)f(Y,A,L)}{f(A,L)^2}\,dY\\
    &= \frac{f(Q_\tau(Y|A,L),A,L)}{f(A,L)} \frac{d}{dt} Q_{\tau,t}(Y|A,L)\Big\lvert_{t=0} + \frac{\mathbbm{1}_{\Tilde{a},\Tilde{l}}(A,L)}{f(A,L)}\int_{-\infty}^{Q_\tau(Y|A,L)}\mathbbm{1}_{\Tilde{y}}(Y)\,dY\\
    & \qquad\qquad\qquad - \frac{\mathbbm{1}_{\Tilde{a},\Tilde{l}}(A,L)}{f(A,L)}\int_{-\infty}^{Q_\tau(Y|A,L)}f(Y|A,L)\,dY\\
    &= \frac{f(Q_\tau(Y|A,L),A,L)}{f(A,L)} \frac{d}{dt} Q_{\tau,t}(Y|A,L)\Big\lvert_{t=0} + \frac{\mathbbm{1}_{\Tilde{a},\Tilde{l}}(A,L)}{f(A,L)} I\{\Tilde{y}\leq Q_\tau(Y|A,L)\} \\
    & \qquad\qquad\qquad - \frac{\mathbbm{1}_{\Tilde{a},\Tilde{l}}(A,L)}{f(A,L)}\tau.
\end{align*}
This leads to
\[
\frac{d}{dt} Q_{\tau,t}(Y|A,L)\Big\lvert_{t=0} = \frac{\mathbbm{1}_{\Tilde{a},\Tilde{l}}(A,L)}{f(Q_\tau(Y|A,L),A,L)}\Big( \tau-I\{\Tilde{y}\leq Q_\tau(Y|A,L)\}\Big).
\]
\end{proof}
\end{lemma}
\begin{lemma}\label{lemma2}
\begin{align*}
    \frac{d}{dt}(A-E_t(A|L))Q_{\tau,t}(Y|A,L)\Big\lvert_{t=0} &= -\frac{\mathbbm{1}_{\Tilde{l}}(L)}{f(L)}(\Tilde{a}-E(A|L))Q_\tau(Y|A,L)\\
    &\qquad\qquad + \{A-E(A|L)\}\mathbbm{1}_{\Tilde{a},\Tilde{l}}(A,L)\frac{\tau-I\{\Tilde{y}\leq Q_\tau(Y|A,L)\}}{f(Q_\tau(Y|A,L),A,L)}
\end{align*}
\end{lemma}
\begin{proof}
    \begin{align*}
    \frac{d}{dt}(A-E_t(A|L))Q_{\tau,t}(Y|A,L)\Big\lvert_{t=0} &= \frac{d}{dt}(A-E_t(A|L))\Big\lvert_{t=0} Q_{\tau}(Y|A,L)\\
    &\qquad\qquad+ (A-E(A|L))\frac{d}{dt}Q_{\tau,t}(Y|A,L)\Big\lvert_{t=0}\\ 
    &=-\frac{\mathbbm{1}_{\Tilde{l}}(L)}{f(L)}(\Tilde{a}-E(A|L))Q_\tau(Y|A,L)\\
    &\qquad\qquad + \{A-E(A|L)\}\mathbbm{1}_{\Tilde{a},\Tilde{l}}(A,L)\frac{\tau-I\{\Tilde{y}\leq Q_\tau(Y|A,L)\}}{f(Q_\tau(Y|A,L),A,L)}
\end{align*}
\end{proof}
\begin{lemma}\label{lemma3}
\begin{align*}
    \frac{d}{dt}(A-E_t(A|L))^2\Big\lvert_{t=0} &= - 2 \{A-E(A|L)\}(\Tilde{a}-E(A|L))\frac{\mathbbm{1}_{\Tilde{l}}(L)}{f(L)}
\end{align*}
\end{lemma}
\begin{proof}
\begin{align*}
    \frac{d}{dt}(A-E_t(A|L))^2\Big\lvert_{t=0} &= - 2 \{A-E(A|L)\}\frac{d}{dt}E_t(A|L)\Big\lvert_{t=0}\\
    &= - 2 \{A-E(A|L)\}(\Tilde{a}-E(A|L))\frac{\mathbbm{1}_{\Tilde{l}}(L)}{f(L)}
\end{align*}
\end{proof}

Taking the derivative of $\Psi_\tau(\mathcal{P}_t)$ with respect to $t$, evaluated at $t=0$, results in
\begin{align*}
    \frac{d}{dt}\Psi_\tau(\mathcal{P}_t)\Big\lvert_{t=0} &= \frac{\frac{d}{dt}E_t[(A-E_t(A|L))Q_{\tau,t}(Y|A,L)]\Big\lvert_{t=0}}{E[\{A-E(A|L)\}^2]}\\
    &\qquad\qquad\qquad\qquad- \frac{\Psi_\tau}{E[\{A-E(A|L)\}^2]}\frac{d}{dt}E_t[(A-E_t(A|L))^2]\Big\lvert_{t=0}\\
    \\
    &= \frac{(\Tilde{a}-E(A|L=\Tilde{l}))Q_\tau(Y|A=\Tilde{a},L=\Tilde{l})}{E[\{A-E(A|L)\}^2]} - \Psi_\tau + \frac{E\left[\frac{d}{dt}(A-E_t(A|L))Q_{\tau,t}(Y|A,L)\Big\lvert_{t=0}\right]}{E[\{A-E(A|L)\}^2]} \\
    &\qquad - \frac{\Psi_\tau}{E[\{A-E(A|L)\}^2]}\Bigg[ (\Tilde{a}-E(A|L=\Tilde{l}))^2 - E[\{A-E(A|L)\}^2]+ E\left(\frac{d}{dt}(A-E_t(A|L))^2\Big\lvert_{t=0}\right) \Bigg]\\
    \\
    &= \frac{1}{E[\{A-E(A|L)\}^2]}\Bigg[ (\Tilde{a}-E(A|L=\Tilde{l}))Q_\tau(Y|A=\Tilde{a},L=\Tilde{l})\\
    &\qquad\qquad- \iint \frac{\mathbbm{1}_{\Tilde{l}}(L)}{f(L)}(\Tilde{a}-E(A|L))Q_\tau(Y|A,L)f(A,L)\,dA\,dL\\
    &\qquad\qquad + \iint \{A-E(A|L)\}\mathbbm{1}_{\Tilde{a},\Tilde{l}}(A,L)\frac{\tau-I\{\Tilde{y}\leq Q_\tau(Y|A,L)\}}{f(Q_\tau(Y|A,L),A,L)}f(A,L)\,dA\,dL\\
    &\qquad\qquad- \Psi_\tau\Bigg( (\Tilde{a}-E(A|L=\Tilde{l}))^2 - 2\iint \{A-E(A|L)\}(\Tilde{a}-E(A|L))\frac{\mathbbm{1}_{\Tilde{l}}(L)}{f(L)}f(A,L)\,dA\,dL \Bigg)\Bigg]\\
    \\
    &= \frac{1}{E[\{A-E(A|L)\}^2]}\Bigg[ (\Tilde{a}-E(A|L=\Tilde{l}))Q_\tau(Y|A=\Tilde{a},L=\Tilde{l})\\
    &\qquad\qquad- (\Tilde{a}-E(A|L=\Tilde{l})) \int Q_\tau(Y|A,L=\Tilde{l})f(A|L=\Tilde{l})\,dA\\
    &\qquad\qquad + (\Tilde{a}-E(A|L=\Tilde{l}))\frac{\tau-I\{\Tilde{y}\leq Q_\tau(Y|A=\Tilde{a},L=\Tilde{l})\}}{f(Q_\tau(Y|A=\Tilde{a},L=\Tilde{l})|A=\Tilde{a},L=\Tilde{l})}\\
    &\qquad\qquad- \Psi_\tau\Bigg( (\Tilde{a}-E(A|L=\Tilde{l}))^2 - 2(\Tilde{a}-E(A|L=\Tilde{l}))\int (A-E(A|L=\Tilde{l}))f(A|L=\Tilde{l})\,dA \Bigg)\Bigg]\\
    \\
    &= \frac{\Tilde{a}-E(A|L=\Tilde{l})}{E[\{A-E(A|L)\}^2]}\Bigg[ Q_\tau(Y|A=\Tilde{a},L=\Tilde{l}) -  E\Big(Q_\tau(Y|A,L=\Tilde{l})\big\lvert L=\Tilde{l}\Big)\\
    &\qquad\qquad\qquad\qquad + \frac{\tau-I\{\Tilde{y}\leq Q_\tau(Y|A=\Tilde{a},L=\Tilde{l})\}}{f(Q_\tau(Y|A=\Tilde{a},L=\Tilde{l})|A=\Tilde{a},L=\Tilde{l})}\\
    &\qquad\qquad\qquad\qquad- \Psi_\tau\Bigg( (\Tilde{a}-E(A|L=\Tilde{l})) - 2E\big( A-E(A|L=\Tilde{l})|L=\Tilde{l}\big) \Bigg)\Bigg]\\
    \\
    &= \frac{\Tilde{a}-E(A|L=\Tilde{l})}{E[\{A-E(A|L)\}^2]}\Bigg[Q_\tau(Y|A=\Tilde{a},L=\Tilde{l}) - E\Big(Q_\tau(Y|A,L=\Tilde{l})\big\lvert L=\Tilde{l}\Big)\\
    &\qquad\qquad\qquad\qquad + \frac{\tau-I\{\Tilde{y}\leq Q_\tau(Y|A=\Tilde{a},L=\Tilde{l})\}}{f(Q_\tau(Y|A=\Tilde{a},L=\Tilde{l})|A=\Tilde{a},L=\Tilde{l})} - \Psi_\tau (\Tilde{a}-E(A|L=\Tilde{l})) \Bigg],
\end{align*}
where we use Lemma \ref{lemma2} and Lemma \ref{lemma3} in the third equation. Going back to the random variables, this results in the following efficient influence function:
\begin{align*}
    \frac{A-E(A|L)}{E[\{A-E(A|L)\}^2]}\Bigg[Q_\tau(Y|A,L) - E\left\{Q_\tau(Y|A,L)|L\right\} + \frac{\tau-I\{Y\leq Q_\tau(Y|A,L)\}}{f\{Q_\tau(Y|A,L)|A,L\}} - \Psi_\tau \{A-E(A|L)\} \Bigg].
\end{align*}

\subsection{Proof of Theorem \ref{theorem2}}
Using a von Mises expansion (see e.g.\ \cite{vdv2000asymptotic}), we have
\begin{align*}
    \hat{\Psi}_\tau-\Psi_\tau &= \frac{1}{n}\sum_{i=1}^n \frac{A_i-E(A_i|L_i)}{E[\{A-E(A|L)\}^2]}\Bigg[Q_\tau(Y_i|A_i,L_i) - E\left\{Q_\tau(Y|A,L)|L\right\}\\
    &\qquad\qquad+ \frac{\tau-I\{Y_i\leq Q_\tau(Y_i|A_i,L_i)\}}{f\{Q_\tau(Y_i|A_i,L_i)|A_i,L_i\}} - \Psi_\tau (A_i-E(A_i|L_i)) \Bigg] + R_1+R_2.
\end{align*}
Here, the first term is $O_p(n^{-1/2})$ and 
\begin{align*}
    R_1 &= \frac{1}{n}\sum_{i=1}^n \Bigg[\frac{A_i-\hat{E}(A_i|L_i)}{\hat{E}[(A_i-\hat{E}(A_i|L_i))^2]}\Bigg\{ \hat{Q}_\tau(Y_i|A_i,L_i) - \hat{E}\Big(\hat{Q}_\tau(Y_i|A_i,L_i)\big\lvert  L_i\Big)\\
    & \qquad\qquad\qquad\qquad\qquad\qquad+ \frac{\tau-I\{Y_i\leq \hat{Q}_\tau(Y_i|A_i,L_i)\}}{\hat{f}\{\hat{Q}_\tau(Y_i|A_i,L_i)|A_i,L_i\}} - \hat{\Psi}_\tau (A_i-\hat{E}(A_i|L_i)) \Bigg\}\\
    & \qquad\qquad\qquad-\frac{A_i-E(A_i|L_i)}{E[\{A-E(A|L)\}^2]}\Bigg\{Q_\tau(Y_i|A_i,L_i) - E\left\{Q_\tau(Y|A,L)|L\right\}\\
    & \qquad\qquad\qquad\qquad\qquad\qquad+ \frac{\tau-I\{Y_i\leq Q_\tau(Y_i|A_i,L_i)\}}{f\{Q_\tau(Y_i|A_i,L_i)|A_i,L_i\}} - \Psi_\tau (A_i-E(A_i|L_i)) \Bigg\}\Bigg]\\
    &\qquad - E\Bigg[\frac{A-\hat{E}(A|L)}{\hat{E}[(A-\hat{E}(A|L))^2]}\Bigg\{ \hat{Q}_\tau(Y|A,L) - \hat{E}\Big(\hat{Q}_\tau(Y|A,L)\big\lvert L\Big)\\
    & \qquad\qquad\qquad\qquad\qquad\qquad+ \frac{\tau-I\{Y\leq \hat{Q}_\tau(Y|A,L)\}}{\hat{f}(\hat{Q}_\tau(Y|A,L)|A,L)} - \hat{\Psi}_\tau (A-\hat{E}(A|L)) \Bigg\}\\
    & \qquad\qquad\qquad-\frac{A-E(A|L)}{E[\{A-E(A|L)\}^2]}\Bigg\{Q_\tau(Y|A,L) - E\left\{Q_\tau(Y|A,L)|L\right\}\\
    & \qquad\qquad\qquad\qquad\qquad\qquad+ \frac{\tau-I\{Y\leq Q_\tau(Y|A,L)\}}{f\{Q_\tau(Y|A,L)|A,L\}} - \Psi_\tau \{A-E(A|L)\} \Bigg\}\Bigg]
\end{align*}
and
\begin{align*}
    R_2 &= \hat{\Psi}_\tau-\Psi_\tau + E\Bigg[\frac{A-\hat{E}(A|L)}{\hat{E}[(A-\hat{E}(A|L))^2]}\Bigg\{ \hat{Q}_\tau(Y|A,L) - \hat{E}\Big(\hat{Q}_\tau(Y|A,L)\big\lvert L\Big)\\
    & \qquad\qquad\qquad\qquad\qquad\qquad+ \frac{\tau-I\{Y\leq \hat{Q}_\tau(Y|A,L)\}}{\hat{f}(\hat{Q}_\tau(Y|A,L)|A,L)} - \hat{\Psi}_\tau (A-\hat{E}(A|L)) \Bigg\}    \Bigg].
\end{align*}

The term $R_1$ can be shown to be $o_p(n^{-1/2})$ by using sample-splitting. A detailed elaboration of this can be found along the lines of the Supplementary Material of \cite{AssumptionLean}. In what follows we take a deeper look into $R_2$ to understand under what conditions it is $o_p(n^{-1/2})$. From the definition of $R_2$, it follows that
\begin{align*}
    R_2 &= \hat{\Psi}_\tau-\Psi_\tau - \hat{\Psi}_\tau\frac{E[(A-\hat{E}(A|L))^2]}{\hat{E}[(A-\hat{E}(A|L))^2]} \\
    &\qquad\qquad + E\Bigg[\frac{A-\hat{E}(A|L)}{\hat{E}[(A-\hat{E}(A|L))^2]}\Bigg\{ \hat{Q}_\tau(Y|A,L) - \hat{E}\Big(\hat{Q}_\tau(Y|A,L)\big\lvert L\Big) + \frac{\tau-I\{Y\leq \hat{Q}_\tau(Y|A,L)\}}{\hat{f}(\hat{Q}_\tau(Y|A,L)|A,L)}  \Bigg\}    \Bigg]\\
    &= (\hat{\Psi}_\tau -\Psi_\tau)\Bigg[ 1- \frac{E[(A-\hat{E}(A|L))^2]}{\hat{E}[(A-\hat{E}(A|L))^2]}\Bigg]\\
    &\qquad\qquad+ E\Bigg[\frac{A-\hat{E}(A|L)}{\hat{E}[(A-\hat{E}(A|L))^2]}\Bigg\{ \hat{Q}_\tau(Y|A,L) - \hat{E}\Big(\hat{Q}_\tau(Y|A,L)\big\lvert L\Big)\\
    & \qquad\qquad\qquad\qquad\qquad\qquad+ \frac{\tau-I\{Y\leq \hat{Q}_\tau(Y|A,L)\}}{\hat{f}(\hat{Q}_\tau(Y|A,L)|A,L)} - \Psi_\tau (A-\hat{E}(A|L)) \Bigg\} \Bigg].\\
\end{align*}
Here, the first term is $o_p(|\hat{\Psi}_\tau -\Psi_\tau|)$, and thus a lower order term. To understand the behaviour of the remaining terms, we make use of
\begin{align*}
    \Psi_\tau &= E\Bigg[\frac{A-E(A|L)}{E[(A-E(A|L))^2]}\Bigg\{Q_\tau(Y|A,L)- E\Big(Q_\tau(Y|A,L)\big\lvert L\Big)+ \frac{\tau-I\{Y\leq Q_\tau(Y|A,L)\}}{f\{Q_\tau(Y|A,L)|A,L\}} \Bigg\}\Bigg].
\end{align*}
To simplify the notation in the following calculations, let $\zeta=Q_\tau(Y|A,L)$. Then the remaining terms can be written as
\begin{align*}
    &\frac{1}{\hat{E}[(A-\hat{E}(A|L))^2]}\Bigg( E\left[(A-\hat{E}(A|L))\Big\{ \hat{\zeta} - \hat{E}(\hat{\zeta}|L)+ \frac{\tau-I\{Y\leq \hat{\zeta}\}}{\hat{f}(\hat{\zeta}|A,L)}\Big\}\right]\\
    &\qquad\qquad - E\left[(A-E(A|L))\Big\{ \zeta - E(\zeta|L)+ \frac{\tau-I\{Y\leq \zeta\}}{f(\zeta|A,L)}\Big\}\right]\frac{E[(A-\hat{E}(A|L))^2]}{E[(A-E(A|L))^2]} \Bigg)\\
    &=\frac{1}{\hat{E}[(A-\hat{E}(A|L))^2]}\Bigg( E\left[(A-\hat{E}(A|L))\Big\{ \hat{\zeta} - \hat{E}(\hat{\zeta}|L)+ \frac{\tau-P(Y\leq \hat{\zeta}|A,L)}{\hat{f}(\hat{\zeta}|A,L)}\Big\}\right]\\
    &\qquad\qquad - E\left[(A-E(A|L))\Big\{ \zeta - E(\zeta|L)+ \frac{\tau-P(Y\leq \zeta|A,L)}{f(\zeta|A,L)}\Big\}\right]\frac{E[(A-\hat{E}(A|L))^2]}{E[(A-E(A|L))^2]} \Bigg)\\
    &=\frac{1}{\hat{E}[(A-\hat{E}(A|L))^2]}\Bigg( E\left[(A-\hat{E}(A|L))\Big\{ \hat{\zeta} - \hat{E}(\hat{\zeta}|L)+ \frac{\tau-P(Y\leq \hat{\zeta}|A,L)}{\hat{f}(\hat{\zeta}|A,L)}\Big\}\right]\\
    &\qquad\qquad - E\left[(A-E(A|L))\Big\{ \zeta - E(\zeta|L)\Big\}\right]\frac{E[(A-\hat{E}(A|L))^2]}{E[(A-E(A|L))^2]} \Bigg).
\end{align*}
Next, using a Taylor expansion around $\tau$, one can write
\[
P(Y\leq \hat{\zeta}|A,L) = \tau + f(\zeta|A,L)(\hat{\zeta}-\zeta)+O_p\{(\hat{\zeta}-\zeta)^2\}.
\]
This reduces the remaining terms to
\begin{align*}
    &\frac{1}{\hat{E}[(A-\hat{E}(A|L))^2]}\Bigg( E\left[(A-\hat{E}(A|L))\Big\{ \hat{\zeta} - \hat{E}(\hat{\zeta}|L)- \frac{f(\zeta|A,L)}{\hat{f}(\hat{\zeta}|A,L)}(\hat{\zeta}-\zeta)+O_p\{(\hat{\zeta}-\zeta)^2\}\Big\}\right]\\
    &\qquad\qquad - E\left[(A-E(A|L))\Big\{ \zeta - E(\zeta|L)\Big\}\right]\frac{E[(A-\hat{E}(A|L))^2]}{E[(A-E(A|L))^2]} \Bigg)\\
    &=\frac{1}{\hat{E}[(A-\hat{E}(A|L))^2]}\Bigg( E\left[(A-\hat{E}(A|L))\Big\{\zeta - \hat{E}(\hat{\zeta}|L)- \left(\frac{f(\zeta|A,L)}{\hat{f}(\hat{\zeta}|A,L)}-1\right)(\hat{\zeta}-\zeta)+O_p\{(\hat{\zeta}-\zeta)^2\}\Big\}\right]\\
    &\qquad\qquad - E\left[(A-E(A|L))\Big\{ \zeta - E(\zeta|L)\Big\}\right]\frac{E[(A-\hat{E}(A|L))^2]}{E[(A-E(A|L))^2]} \Bigg)\\
    &=\frac{1}{\hat{E}[(A-\hat{E}(A|L))^2]}\Bigg( E\left[(A-E(A|L))\Big\{\zeta - \hat{E}(\hat{\zeta}|L)- \left(\frac{f(\zeta|A,L)}{\hat{f}(\hat{\zeta}|A,L)}-1\right)(\hat{\zeta}-\zeta)+O_p\{(\hat{\zeta}-\zeta)^2\}\Big\}\right]\\
    &\qquad\qquad + E\left[(E(A|L)-\hat{E}(A|L))\Big\{\zeta - \hat{E}(\hat{\zeta}|L)- \left(\frac{f(\zeta|A,L)}{\hat{f}(\hat{\zeta}|A,L)}-1\right)(\hat{\zeta}-\zeta)+O_p\{(\hat{\zeta}-\zeta)^2\}\Big\}\right]\\
    &\qquad\qquad - E\left[(A-E(A|L))\Big\{ \zeta - E(\zeta|L)\Big\}\right]\\
    &\qquad\qquad + \left(1-\frac{E[(A-\hat{E}(A|L))^2]}{E[(A-E(A|L))^2]}\right)E\left[(A-E(A|L))\Big\{ \zeta - E(\zeta|L)\Big\}\right] \Bigg)\\
    &=\frac{1}{\hat{E}[(A-\hat{E}(A|L))^2]}\Bigg( E\left[(A-E(A|L))\Big\{ - \hat{E}(\hat{\zeta}|L)- \left(\frac{f(\zeta|A,L)}{\hat{f}(\hat{\zeta}|A,L)}-1\right)(\hat{\zeta}-\zeta)+O_p\{(\hat{\zeta}-\zeta)^2\}\Big\}\right]\\
    &\qquad\qquad + E\left[(E(A|L)-\hat{E}(A|L))\Big\{\zeta - \hat{E}(\hat{\zeta}|L)- \left(\frac{f(\zeta|A,L)}{\hat{f}(\hat{\zeta}|A,L)}-1\right)(\hat{\zeta}-\zeta)+O_p\{(\hat{\zeta}-\zeta)^2\}\Big\}\right]\\
    &\qquad\qquad - E\left[(A-E(A|L))\Big\{ - E(\zeta|L)\Big\}\right]\\
    &\qquad\qquad + \left(1-\frac{E[(A-\hat{E}(A|L))^2]}{E[(A-E(A|L))^2]}\right)E\left[(A-E(A|L))\Big\{ \zeta - E(\zeta|L)\Big\}\right] \Bigg)\\
    &=\frac{1}{\hat{E}[(A-\hat{E}(A|L))^2]}\Bigg( E\left[(A-E(A|L))\Big\{O_p\{(\hat{\zeta}-\zeta)^2\}- \left(\frac{f(\zeta|A,L)}{\hat{f}(\hat{\zeta}|A,L)}-1\right)(\hat{\zeta}-\zeta)\Big\}\right]\\
    &\qquad + E\left[(E(A|L)-\hat{E}(A|L))\Big\{\zeta - \hat{E}(\hat{\zeta}|L)+E(\zeta|L)-E(\zeta|L)- \left(\frac{f(\zeta|A,L)}{\hat{f}(\hat{\zeta}|A,L)}-1\right)(\hat{\zeta}-\zeta)+O_p\{(\hat{\zeta}-\zeta)^2\}\Big\}\right]\\
    &\qquad\qquad + \left(1-\frac{E[(A-\hat{E}(A|L))^2]}{E[(A-E(A|L))^2]}\right)E\left[(A-E(A|L))\Big\{ \zeta - E(\zeta|L)\Big\}\right] \Bigg)\\
    &=\frac{1}{\hat{E}[(A-\hat{E}(A|L))^2]}\Bigg( E\left[(A-E(A|L))\Big\{O_p\{(\hat{\zeta}-\zeta)^2\}- \left(\frac{f(\zeta|A,L)}{\hat{f}(\hat{\zeta}|A,L)}-1\right)(\hat{\zeta}-\zeta)\Big\}\right]\\
    &\qquad\qquad + E\left[(E(A|L)-\hat{E}(A|L))\Big\{\zeta -E(\zeta|L)- \left(\frac{f(\zeta|A,L)}{\hat{f}(\hat{\zeta}|A,L)}-1\right)(\hat{\zeta}-\zeta)+O_p\{(\hat{\zeta}-\zeta)^2\}\Big\}\right]\\
    &\qquad\qquad + E\left[(E(A|L)-\hat{E}(A|L))\Big\{E(\zeta|L)- \hat{E}(\hat{\zeta}|L)\Big\}\right]\\
    &\qquad\qquad + \left(1-\frac{E[(A-\hat{E}(A|L))^2]}{E[(A-E(A|L))^2]}\right)E\left[(A-E(A|L))\Big\{ \zeta - E(\zeta|L)\Big\}\right] \Bigg)\\
    &=\frac{1}{\hat{E}[(A-\hat{E}(A|L))^2]}\Bigg( E\left[(A-E(A|L))\Big\{O_p\{(\hat{\zeta}-\zeta)^2\}- \left(\frac{f(\zeta|A,L)}{\hat{f}(\hat{\zeta}|A,L)}-1\right)(\hat{\zeta}-\zeta)\Big\}\right]\\
    &\qquad\qquad + E\left[(E(A|L)-\hat{E}(A|L))\Big\{O_p\{(\hat{\zeta}-\zeta)^2\}- \left(\frac{f(\zeta|A,L)}{\hat{f}(\hat{\zeta}|A,L)}-1\right)(\hat{\zeta}-\zeta)\Big\}\right]\\
    &\qquad\qquad + E\left[(E(A|L)-\hat{E}(A|L))\Big\{E(\zeta|L)- \hat{E}(\hat{\zeta}|L)\Big\}\right]\\
    &\qquad\qquad + \left(1-\frac{E[(A-E(A|L)+E(A|L)-\hat{E}(A|L))^2]}{E[(A-E(A|L))^2]}\right)E\left[(A-E(A|L))\Big\{ \zeta - E(\zeta|L)\Big\}\right] \Bigg)\\
    &=\frac{1}{\hat{E}[(A-\hat{E}(A|L))^2]}\Bigg( E\left[(A-\hat{E}(A|L))\Big\{O_p\{(\hat{\zeta}-\zeta)^2\}- \left(\frac{f(\zeta|A,L)}{\hat{f}(\hat{\zeta}|A,L)}-1\right)(\hat{\zeta}-\zeta)\Big\}\right]\\
    &\qquad\qquad + E\left[(E(A|L)-\hat{E}(A|L))\Big\{E(\zeta|L)- \hat{E}(\hat{\zeta}|L)\Big\}\right]\\
    &\qquad\qquad - \frac{E[2(A-E(A|L))(E(A|L)-\hat{E}(A|L))+(E(A|L)-\hat{E}(A|L))^2]}{E[(A-E(A|L))^2]}E\left[(A-E(A|L))\Big\{ \zeta - E(\zeta|L)\Big\}\right] \Bigg)\\
    &=\frac{1}{\hat{E}[(A-\hat{E}(A|L))^2]}\Bigg( E\left[(A-\hat{E}(A|L))O_p\{(\hat{\zeta}-\zeta)^2\}\right] + E\left[(A-\hat{E}(A|L))\left(1-\frac{f(\zeta|A,L)}{\hat{f}(\hat{\zeta}|A,L)}\right)(\hat{\zeta}-\zeta)\right]\\
    &\qquad\qquad + E\left[(E(A|L)-\hat{E}(A|L))\Big\{E(\zeta|L)- \hat{E}(\hat{\zeta}|L)\Big\}\right]\\
    &\qquad\qquad - \frac{E[(E(A|L)-\hat{E}(A|L))^2]}{E[(A-E(A|L))^2]}E\left[(A-E(A|L))\Big\{ \zeta - E(\zeta|L)\Big\}\right] \Bigg)\\
    &=\frac{1}{\hat{E}[(A-\hat{E}(A|L))^2]}\Bigg( E\left[(A-\hat{E}(A|L))O_p\{(\hat{\zeta}-\zeta)^2\}\right] + E\left[(A-\hat{E}(A|L))\left(1-\frac{f(\zeta|A,L)}{\hat{f}(\hat{\zeta}|A,L)}\right)(\hat{\zeta}-\zeta)\right]\\
    &\qquad\qquad + E\left[(E(A|L)-\hat{E}(A|L))\Big\{E(\zeta|L)- \hat{E}(\hat{\zeta}|L)\Big\}\right] - E[(E(A|L)-\hat{E}(A|L))^2]\Psi_\tau \Bigg).
\end{align*}

Assuming $A-\hat{E}(A|L) = O_p(1)$ and using the Cauchy-Schwarz inequality, the term $R_2$ is $o_p(n^{-1/2})$ under the following conditions:
\begin{align*}
    E\left[ (\hat{Q}_\tau(Y|A,L)-Q_\tau(Y|A,L))^2 \right] &= o_p(n^{-1/2}),\\
    E\left[ \left(1-\frac{f\{Q_\tau(Y|A,L)|A,L\}}{\hat{f}(\hat{Q}_\tau(Y|A,L)|A,L)}\right)^2 \right]^{1/2}E\left[ (\hat{Q}_\tau(Y|A,L)-Q_\tau(Y|A,L))^2 \right]^{1/2} &= o_p(n^{-1/2}),\\
    E\left[(E(A|L)-\hat{E}(A|L))^2\right]^{1/2}E\left[ \left(E(Q_\tau(Y|A,L)|L)-\hat{E}(\hat{Q}_\tau(Y|A,L)|L)\right)^2 \right]^{1/2} &= o_p(n^{-1/2}),\\
    E\left[(E(A|L)-\hat{E}(A|L))^2\right]&= o_p(n^{-1/2}).
\end{align*}
Note that if $\Psi_\tau=0$, the final condition is redundant.

\subsection{Sketch of a proof that the sum in (\ref{tmlefun}) is $o_P(n^{-1/2})$}
Here, we give a sketch of a proof that
\begin{align*}
\frac{1}{n}\sum_{i=1}^n\left(A_i-\hat{E}(A_i|L_i)\right)\Bigg[ \frac{\tau-I\{Y_i\leq \widetilde{Q}_\tau(Y_i|A_i,L_i)\}}{\hat{f}(\widetilde{Q}_\tau(Y_i|A_i,L_i)|A_i,L_i)} \Bigg] = o_P(n^{-1/2}).
\end{align*}
We consider
\[
\widetilde{Q}_\tau(Y|A,L) = \hat{Q}_\tau(Y|A,L) + \epsilon  w(\hat{Q}_{\tau,i}),
\]
with $w(\hat{Q}_{\tau,i})$ defined as in (\ref{weight}). Thus, the aim is to prove that
\begin{align}\label{appD:tmle}
\frac{1}{n}\sum_{i=1}^nw(\hat{Q}_{\tau,i})\left[ \tau-I\{Y_i\leq \hat{Q}_\tau(Y_i|A_i,L_i) + \epsilon  w(\hat{Q}_{\tau,i})\} \right]
\end{align}
is $o_P(n^{-1/2})$.

One can control the sign of (\ref{appD:tmle}) by changing the value of $\epsilon$. Indeed, for $\epsilon\to + \infty$ and $w(\hat{Q}_{\tau,i})\geq 0$, we have 
\[w(\hat{Q}_{\tau,i})\left[ \tau-I\{Y_i\leq \hat{Q}_\tau(Y_i|A_i,L_i) + \epsilon  w(\hat{Q}_{\tau,i})\} \right] = w(\hat{Q}_{\tau,i})(\tau-1)\leq 0.
\] 
For $\epsilon\to + \infty$ and $w(\hat{Q}_{\tau,i})\leq 0$, we have \[w(\hat{Q}_{\tau,i})\left[ \tau-I\{Y_i\leq \hat{Q}_\tau(Y_i|A_i,L_i) + \epsilon  w(\hat{Q}_{\tau,i})\} \right] = w(\hat{Q}_{\tau,i})\tau\leq 0,\] 
such that 
\[\lim_{\epsilon\to +\infty}\frac{1}{n}\sum_{i=1}^nw(\hat{Q}_{\tau,i})\left[ \tau-I\{Y_i\leq \hat{Q}_\tau(Y_i|A_i,L_i) + \epsilon  w(\hat{Q}_{\tau,i})\} \right]\leq 0.\]
Analogously, one can show
\[\lim_{\epsilon\to -\infty}\frac{1}{n}\sum_{i=1}^nw(\hat{Q}_{\tau,i})\left[ \tau-I\{Y_i\leq \hat{Q}_\tau(Y_i|A_i,L_i) + \epsilon  w(\hat{Q}_{\tau,i})\} \right]\geq 0.\]
By slightly adjusting the value of $\epsilon$, terms from the summation can change from 
$w(\hat{Q}_{\tau,i})(\tau-1)$ to $w(\hat{Q}_{\tau,i})\tau$ or vice versa. Changing one term leads to a change of (\ref{appD:tmle}) of magnitude $\frac{w(\hat{Q}_{\tau,i})}{n}$. This implies that we can get (\ref{appD:tmle}) within a distance $\lvert\frac{w(\hat{Q}_{\tau,i})}{n}\lvert$ of 0. If $w(\hat{Q}_{\tau,i})=O_P(1)$, which holds if $A-\hat{E}(A|L)=O_p(1)$ and $\frac{1}{\hat{f}(Q_\tau(Y|A,L)|A,L)}>\sigma$ with probability 1 for some $\sigma>0$, we have that (\ref{appD:tmle}) is $O_P(1/n)$. Potentially, multiple terms change at the same time. If $M$ terms change at the same time, we have that (\ref{appD:tmle}) lies within a distance $M\lvert\frac{w(\hat{Q}_{\tau,i})}{n}\lvert$ of 0. Consequently, (\ref{appD:tmle}) is $O_P(M/n)$. If $M=o_p(\sqrt{n})$, we have that (\ref{appD:tmle}) is $o_P(n^{-1/2})$.

\newpage
\section{Appendix C}\label{appC}
\subsection{More details about the construction of estimator (\ref{TMLEvs})}
Assume that a (stepwise) variable selection estimator for $Q_\tau(Y|A,L)$ has the following form: 
\[\hat{Q}_\tau(Y|A,L) = \hat{\beta}_\tau A+ \hat{\alpha}^TL^*,\] 
with $L^*$ consisting of a subset of variables of $L$. Retargeting
\begin{equation*}
    \widetilde{Q}_\tau(Y\mid A,L) = \hat{Q}_\tau(Y\mid A,L) + \hat{\epsilon}  w(\hat{Q}_{\tau,i}),
\end{equation*} then leads to:
\[
\widetilde{Q}_\tau(Y|A,L) = \hat{\beta} A+ \hat{\alpha}^TL^* + \epsilon \frac{A-\hat{E}(A|L)}{\hat{f}(\hat{Q}_\tau(Y|A,L)|A,L)}.
\]
Then, one can obtain
\begin{align*}
    &\widetilde{Q}_\tau(Y|A,L) - E\Big(\widetilde{Q}_\tau(Y|A,L)\big\lvert L\Big)\\
    &= \hat{\beta} A+ \hat{\alpha}^TL^* + \epsilon \frac{A-\hat{E}(A|L)}{\hat{f}(\hat{Q}_\tau(Y|A,L)|A,L)} - E\left(\hat{\beta} A+ \hat{\alpha}^TL^* + \epsilon \frac{A-\hat{E}(A|L)}{\hat{f}(\hat{Q}_\tau(Y|A,L)|A,L)}\big\lvert L\right)\\
    &= \hat{\beta}(A-\hat{E}(A|L)) + \epsilon\left\{\frac{A-\hat{E}(A|L)}{\hat{f}(\widetilde{Q}_\tau(Y|A,L)|A,L)} - E\left(\frac{A-\hat{E}(A|L)}{\hat{f}(\widetilde{Q}_\tau(Y|A,L)|A,L)}\right)\right\}
\end{align*}
which allows us to construct a TMLE estimator as follows:
\begin{align*}
    \hat{\Psi}_\tau^{TMLE} &= \frac{1}{n}\sum_{i=1}^n\frac{A_i-\hat{E}(A_i|L_i)}{\frac{1}{n}\sum_{i=1}^n(A_i-\hat{E}(A_i|L_i))^2}\Bigg[\widetilde{Q}_\tau(Y_i|A_i,L_i)\\
    &\qquad\qquad- \hat{E}\Big(\widetilde{Q}_\tau(Y_i|A_i,L_i)\big\lvert  L_i\Big) + \frac{\tau-I\{Y_i\leq \widetilde{Q}_\tau(Y_i|A_i,L_i)\}}{\hat{f}(\widetilde{Q}_\tau(Y_i|A_i,L_i)|A_i,L_i)} \Bigg]\\
    &= \hat{\beta} + \frac{1}{n}\sum_{i=1}^n\frac{A_i-\hat{E}(A_i|L_i)}{\frac{1}{n}\sum_{i=1}^n(A_i-\hat{E}(A_i|L_i))^2}\Bigg[\epsilon\Bigg\{\frac{A_i-\hat{E}(A_i|L_i)}{\hat{f}(\widetilde{Q}_\tau(Y_i|A_i,L_i)|A_i,L_i)}\\
    &\qquad-\hat{E}\left(\frac{A_i-\hat{E}(A_i|L_i)}{\hat{f}(\widetilde{Q}_\tau(Y_i|A_i,L_i)|A_i,L_i)}\big\lvert L_i\right)\Bigg\}
    + \frac{\tau-I\{Y_i\leq \widetilde{Q}_\tau(Y_i|A_i,L_i)\}}{\hat{f}(\widetilde{Q}_\tau(Y_i|A_i,L_i)|A_i,L_i)} \Bigg].
\end{align*}
Similarly to (\ref{TMLE}), the last term will be approximately zero.

\newpage
\section{Appendix D}\label{appD}
\subsection{Additional simulation results for experiment 1}\label{app: addsimres}
Experiment 1 was partially carried out with the computationally more expensive `haldensify' estimator for the conditional density of $Y$ given $A$ and $L$. In setting 3 we did not consider a TMLE estimator because this was computationally too heavy. The results are shown in Table \ref{tab:hal}.

In general, the DML estimator performs better when using this conditional density estimator. However, the TMLE estimator performs better when using the simpler `FKSUM' density estimator, especially for more extreme quantiles.
\begin{table}[h]
\centering
\addtolength{\leftskip}{-2cm}
\addtolength{\rightskip}{-2cm}
\begin{tabular}{clcccccccccccccc}
 &  &   \multicolumn{4}{c}{$\tau = 0.5$} &  & \multicolumn{4}{c}{$\tau = 0.75$} &  & \multicolumn{4}{c}{$\tau = 0.9$} \\
\multirow{-2}{*}{\textbf{Setting}} & \multirow{-2}{*}{\textbf{estimator}}  & \textbf{bias} & \textbf{SD} & \textbf{SE} & \textbf{Cov} &  & \textbf{bias} & \textbf{SD} & \textbf{SE} & \textbf{Cov} &  & \textbf{bias} & \textbf{SD} & \textbf{SE} & \textbf{Cov} \\
 & Oracle   & 0.21 & 19 & 20 & 97.4 &  & -1.2 & 33 & 35 & 95.3 &  & -2.9 & 57 & 60 & 95.7 \\
 & Plug-in   & -70 & 13 & 1.5 & 0.1 &  & -73 & 14 & 1.8 & 0.1 &  & -63 & 23 & 3.5 & 0.2 \\
 & DML   & -3.6 & 25 & 22 & 91.0 &  & -6.6 & 42 & 32 & 83.6 &  & -32 & 40 & 35 & 79.5 \\
 & DML-CF   & -5.9 & 22 & 21 & 91.9 &  & -17 & 33 & 30 & 86.8 &  & -40 & 36 & 25 & 59.2 \\
 & TMLE   & -39 & 25 & 23 & 58.8 &  & -69 & 41 & 33 & 44.7 &  & -140 & 66 & 45 & 15.1 \\
 \multirow{-6}{*}{1} & TMLE-CF &  -1.4 & 24 & 24 & 95.3 &  & -1.4 & 39 & 30 & 85.6 &  & 15 & 76 & 24 & 47.8 \\
 & Oracle   & 0.044 & 25 & 26 & 94.6 &  & -0.025 & 44 & 45 & 94.6 &  & -4.0 & 74 & 79 & 95.3 \\
 & Plug-in   & -87 & 27 & 2.8 & 0.5 &  & -120 & 40 & 4.1 & 1.0 &  & -150 & 59 & 6.8 & 1.3 \\
 & DML & 3 & 32 & 27 & 90.7 &  & -11 & 55 & 39 & 78.5 &  & -97 & 64 & 42 & 39.3 \\
 & DML-CF   & -5.0 & 28 & 24 & 88.0 &  & -34 & 43 & 33 & 74.0 &  & -110 & 60 & 26 & 15.2 \\
 & TMLE   & -39 & 30 & 29 & 70.0 &  & -92 & 57 & 43 & 44.1 &  & -210 & 72 & 51 & 6.3 \\
\multirow{-6}{*}{2} & TMLE-CF & 7.0 & 31 & 28 & 91.6 &  & 6.5 & 49 & 33 & 82.3 &  & 32 & 89 & 25 & 40.9 \\
 & Oracle   & 0.19 & 3.4 & 3.6 & 94.6 &  & -0.11 & 6.2 & 6.3 & 94.0 &  & -0.30 & 10 & 11 & 94.9 \\
 & Plug-in   & -17 & 6.4 & 1.6 & 1.5 &  & -25 & 8.2 & 1.9 & 0.5 &  & -39 & 11 & 2.1 & 0.3 \\
 & DML   & -5.2 & 7.1 & 4.2 & 63.4 &  & -8.1 & 9.3 & 5.8 & 60.8 &  & -33 & 11 & 4.5 & 2.6 \\
\multirow{-4}{*}{3} & DML-CF   & -4.8 & 4.0 & 3.4 & 66.7 &  & -13 & 5.7 & 3.4 & 13.4 &  & -2.8 & 8.1 & 3.1 & 0.5
\end{tabular}
\caption{Simulation results for experiment 1 with density estimated via the `haldensify' R package: sample size 500, Monte Carlo bias, Monte Carlo standard deviation (SD), average of the influence function based standard errors (SE) and coverage of 95\% Wald confidence intervals (Cov). All values have been multiplied by $10^2$.}
\label{tab:hal}
\end{table}

Table \ref{tab:tmleterm} shows a summary of the term we tried to make zero with targeted learning (\ref{tmlefun}). When using cross-fitting, we observe that this term is very close to zero, averaging over all simulation runs.
\begin{table}[h]
\centering
\begin{tabular}{cccc}
 &  & TMLE & TMLE-CF \\
Minimum &  & -0.22 & -0.059 \\
1st Quartile && -0.087 & -0.0029\\
Median &  & -0.0028 & 0.00013 \\
Mean &  & -0.036 & 0.00023 \\
3rd Quartile && 0.00013 & 0.0033\\
Maximum &  & 0.20 & 0.064
\end{tabular}
\caption{Summary of the term we wish to make zero with targeted learning (\ref{tmlefun}) for the homoscedastic setting in experiment 1.}
\label{tab:tmleterm}
\end{table}

\begin{table}[h]
\centering
\addtolength{\leftskip}{-2cm}
\addtolength{\rightskip}{-2cm}
\begin{tabular}{ccccccc}
\textbf{Quantile} & \textbf{estimator} &  & \textbf{bias} & \textbf{SD} & \textbf{SE} & \textbf{Cov} \\
\multirow{2}{*}{$0.5$} & TMLE.1 && -0.38 & 0.25 & 0.19 & 44.6 \\
 & TMLE.1-CF && 0.082 & 0.27 & 0.18 & 83.0\\
\multirow{2}{*}{$0.75$}& TMLE.1 && -0.71 & 0.34 & 0.25 & 25.2 \\
 & TMLE.1-CF && 0.10 & 0.44 & 0.25 & 73.7 \\
\multirow{2}{*}{$0.9$} & TMLE.1 && -1.1 & 0.47 & 0.27 & 14.0 \\
 & TMLE.1-CF && 0.16 & 0.75 & 0.34 & 62.2
\end{tabular}
\caption{Simulation results for the homoscedasticity setting in experiment 1 with the 1-step TMLE estimator: sample size $n=500$, quantile $\tau$, Monte Carlo bias, Monte Carlo standard deviation (SD), average of the influence function based standard errors (SE) and coverage of 95\% Wald confidence intervals (Cov).}
\label{table1step}
\end{table}

The TMLE(-CF) estimator for continuous exposure was adjusted such that only 1 targeting step was carried out. For binary exposure, this step was repeated until convergence. In Table \ref{table1step}, we show the results of the 1-step TMLE estimator for binary exposure (experiment 1, setting 1). Comparing these results with the results in Table \ref{table1}, it can be observed that 1-step targeting does not achieve the same performance as targeting until convergence.

\subsection{Additional simulation results for experiment 4}
Table \ref{tab: exp4app} contains the results for the following estimators for the data generating mechanism described in experiment 4: the oracle estimator (Oracle), the naive plug-in (Plug-in) estimator, the debiased machine learning estimator (\ref{DML}) based on quantile random forests without cross-fitting (DML) and with 5-fold cross-fitting (DML-CF), the targeted learning estimator (\ref{TMLE}) without cross-fitting (TMLE) and with 5-fold cross-fitting (TMLE-CF). We furthermore included some adjusted estimators. One is obtained by performing quantile regression with AIC-based backward stepwise variable selection (QRvs). Another one is the debiased machine learning estimator defined in (\ref{DMLvs}) where quantile random forests are replaced by quantile regression with AIC-based backward stepwise variable selection, once without cross-fitting (DMLvs) and once with 5-fold cross-fitting (DMLvs-CF). The third one is the TMLE estimator defined in (\ref{TMLEvs}), once without cross-fitting (TMLEvs) and once with 5-fold cross-fitting (TMLEvs-CF).

\begin{table}[h]
\centering
\addtolength{\leftskip}{-2cm}
\addtolength{\rightskip}{-2cm}
\begin{tabular}{cccccccccccc}
& & & \multicolumn{4}{c}{$\mathbf{n=250}$} &  & \multicolumn{4}{c}{$\mathbf{n=500}$} \\ 
\multirow{-2}{*}{\textbf{Quantile}} & \multirow{-2}{*}{\textbf{estimator}} &  & \textbf{bias} & \textbf{SD} & \textbf{SE} & \textbf{Cov} &  & \multicolumn{1}{c}{\textbf{bias}} & \multicolumn{1}{c}{\textbf{SD}} & \multicolumn{1}{c}{\textbf{SE}} & \multicolumn{1}{c}{\textbf{Cov}} \\
\multirow{11}{*}{$\tau=0.5$} & Oracle && 0.66 & 15 & 16 & 96.0 &  &  0.25 & 10 & 11 & 95.7     \\
 & QRvs && 2.4 & 17 & 16 & 91.9 &  &  0.36 & 12 & 11 & 92.8      \\
 & Plug-in && -36 & 11 & 3.6 & 0.9 &  &  -23 & 8.6 & 2.6 & 2.7       \\
 & DML && -12 & 20 & 20 & 90.1 &  &  -20 & 13 & 12 &   59.4    \\
 & DML-CF && 6.6 & 16 & 15 & 90.9 &  &  4.0 & 10 & 10 &   92.1     \\
 & DMLvs && 2.0 & 18 & 13 & 81.8 &  &  0.33 & 12 & 9.9 &   90.2   \\
 & DMLvs-CF && 1.6 & 17 & 9.8 & 74.0 &  &  0.043 & 12 & 8.5 &   84.1     \\
  & TMLE && -33 & 18 & 20 & 61.0 &  &  -26 & 15 & 12 &   40.0    \\
 & TMLE-CF && 19 & 20 & 22 & 87.7 &  &  8.5 & 13 & 15 &   92.8     \\
 & TMLEvs && -29 & 25 & 13 & 67.5 &  &  -0.49 & 15 & 10 &   81.7   \\
 & TMLEvs-CF && 1.3 & 20 & 19 & 94.3 &  &  -0.32 & 13 & 12 &   94.8     \\
\multirow{11}{*}{$\tau=0.75$} & Oracle && -0.016 & 16 & 18 & 96.7 &  &  -0.54 & 12 & 12 &   95.0   \\
 & QRvs && -1.4 & 19 & 17 & 91.3 &  &  0.0095 & 13 & 12 &  93.2     \\
 & Plug-in && -42 & 12 & 4.1 & 1.1 &  &  -29 & 9.1 & 2.9 & 1.4     \\
 & DML && -19 & 20 & 19 & 77.8 &  &  -22 & 13 & 11 & 47.3    \\
 & DML-CF && -0.44 & 18 & 17 & 94.4 &  &  -2.3 & 12 & 11 & 92.3     \\
 & DMLvs && -0.78 & 20 & 12 & 77.2 &  &  -0.24 & 13 & 9.9 &  86.2    \\
 & DMLvs-CF && -0.90 & 18 & 10 & 72.7 &  &  -0.21 & 12 & 8.8 &  82.5    \\
  & TMLE && -41 & 21 & 20 & 45.0 &  &  -37 & 16 & 12 &   19.9    \\
 & TMLE-CF && 16 & 24 & 24 & 88.5 &  &  5.7 & 16 & 16 &   93.7     \\
 & TMLEvs && -3.9 & 28 & 13 & 57.4 &  &  -1.1 & 17 & 9.9 &   72.5   \\
 & TMLEvs-CF && 0.52 & 21 & 21 & 94.0 &  &  -0.28 & 14 & 13 &   93.5     \\
\multirow{11}{*}{$\tau=0.9$} & Oracle && -0.88 & 21 & 22 & 95.7 &  &  0.39 & 14 & 15 &  95.8  \\
 & QRvs && 0.065 & 23 & 20 & 89.7 &  &  0.094 & 16 & 14 & 90.8      \\
 & Plug-in && -64 & 8.4 & 037 & 0.0 &  &  -50 & 7.5 & 3.1 & 0.0      \\
 & DML && -50 & 16 & 19 & 27.3 &  &  -46 & 9.1 & 9.1 &  0.8     \\
 & DML-CF && -9.0 & 25 & 26 & 93.0 &  &  -14 & 16 & 16 & 84.8 \\
 & DMLvs && -0.52 & 25 & 11 & 60.9 &  &  -0.23 & 17 & 9.3 &  73.8    \\
 & DMLvs-CF && 0.41 & 20 & 11 & 73.0 &  &  0.021 & 14 & 9.5 &   80.7  \\
  & TMLE && -67 & 23 & 23 & 18.1 &  &  -67 & 16 & 13 &   3.1    \\
 & TMLE-CF && 21 & 38 & 33 & 83.1 &  &  4.2 & 27 & 22 &   88.1     \\
 & TMLEvs && -6.1 & 41 & 11 & 30.4 &  &  -1.2 & 26 & 9.5 &   42.8   \\
 & TMLEvs-CF && 0.21 & 24 & 24 & 94.4 &  &  -0.18 & 16 & 16 &   93.7  
\end{tabular}
\caption{Simulation results for variable selection: sample size $n$, quantile $\tau$, Monte Carlo bias, Monte Carlo standard deviation (SD), average of the influence function based standard errors (SE) and coverage of 95\% Wald confidence intervals (Cov). All values have been multiplied by $10^2$.}
\label{tab: exp4app}
\end{table}    

\clearpage
\subsection{Propensity score comparison: experiment 1 vs.\ 2}\label{app: propsc}
To get some insight into the propensity scores used in experiment 2, histograms are provided below. The left histogram shows a histogram of the propensity scores used in experiment 1 (with binary $A$). On the right, a histogram of the propensity scores from experiment 2 is shown. Each time, a sample of size 100,000 was generated according to the two data-generating mechanisms. Then a histogram of the propensity scores was constructed. It is clear that more extreme propensity scores occur in experiment 2.
\begin{figure}[h]
     \centering
     \begin{subfigure}[h]{0.45\textwidth}
         \centering
         \includegraphics[width=\textwidth]{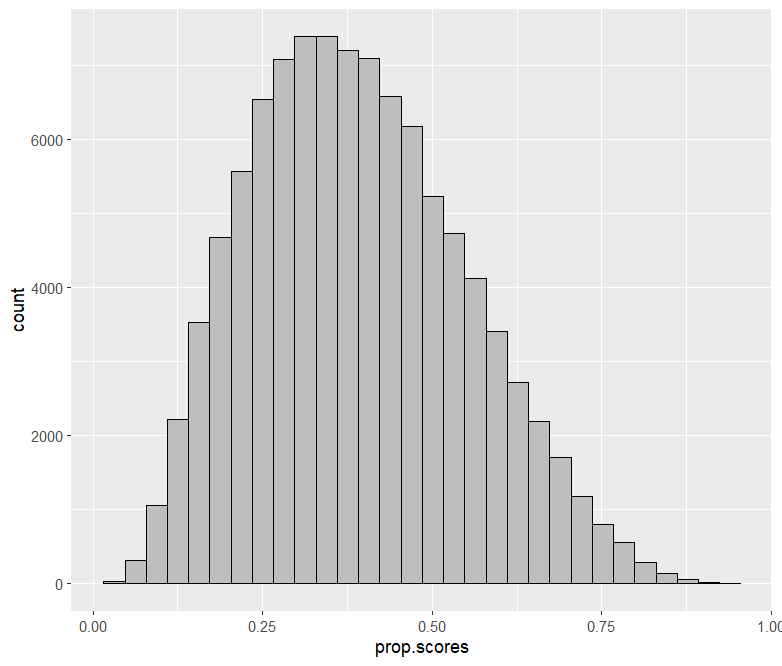}
     \end{subfigure}
     \hfill
     \begin{subfigure}[h]{0.45\textwidth}
         \centering
         \includegraphics[width=\textwidth]{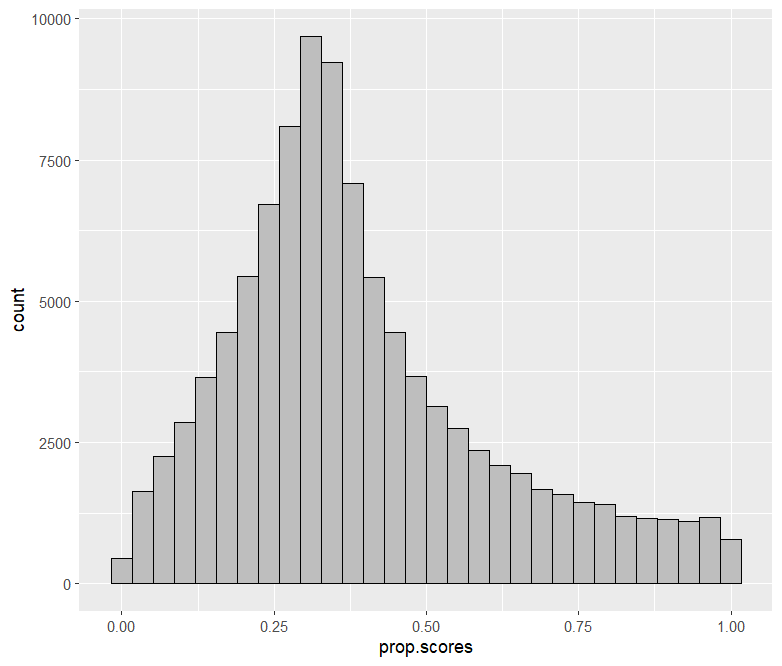}
     \end{subfigure}
        \caption{Propensity score histograms corresponding to the binary exposure settings of experiment 1 (left) and experiment 2 (right).}
        \label{fig:histograms}
\end{figure}

\newpage
\section{Appendix E}\label{appE}
\subsection{Sensitivity analysis for the data analysis}
As shown by the simulation studies, the uncertainty of the proposed TMLE estimator is sometimes underestimated for extreme quantiles. Depending on how the folds in the cross-fitting procedure are chosen, the results may differ. We repeated the analysis for the 90th percentile 10 times with different seeds, once with 5-fold cross-fitting and once with 10-fold cross-fitting. The results are shown in Fig.\ \ref{fig:sens}
\begin{figure}[h]
    \centering
    \includegraphics[width=0.8\textwidth]{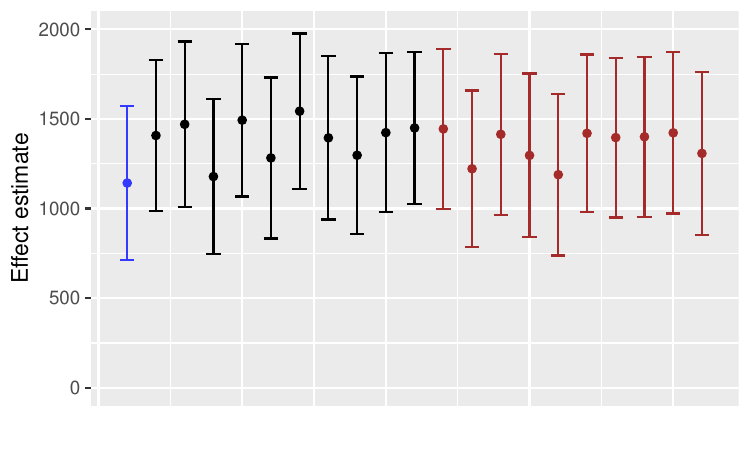}
    \caption{Repeating the analysis 10 times for the 90th percentile, once with 5 folds (shown in black), once with 10 folds (shown in brown). The blue interval corresponds to the one reported in Section \ref{sec:data}.}
    \label{fig:sens}
\end{figure}

\newpage
\section{Appendix F}\label{appF}
\subsection{Generalization to link function}
There are settings in which it could be useful to generalize the estimand (\ref{estimand3}) (or, equivalently, (\ref{estimand2}) so that it includes a link function. A realistic example is a survival analysis where the outcome can only take positive values (e.g.\ when it is assumed to be exponentially distributed with exponential hazard). The quantile in our estimand could take negative values, which makes no sense in this setting. It would be reasonable here to use a link function and replace model (\ref{model}) by
\begin{equation}\label{linkmodel}
    g\{Q_\tau(Y|A,L)\} = \beta_\tau A + \omega_\tau (L),
\end{equation}
with link function $g$ equal to the natural logarithm. This straightforwardly leads to the following estimand for $\beta_\tau$:
\begin{equation}\label{linkestimand}
    \Theta_\tau = \frac{E[\{A-E(A|L)\}(g\{Q_\tau(Y|A,L)\}-E[g\{Q_\tau(Y|A,L)\}|L])]}{E[\{A-E(A|L)\}^2]}.
\end{equation}
For this estimand, the efficient influence function can be calculated in the same way as in the proof of Theorem \ref{theorem1}. The main difference occurs in Lemma \ref{lemma2}. With link function, by making use of the chain rule one would obtain
\begin{align*}
    &\frac{d}{dt}(A-E_t(A|L))g\{Q_{\tau,t}(Y|A,L)\}\Big\lvert_{t=0}\\
    &= \frac{d}{dt}(A-E_t(A|L))\Big\lvert_{t=0} g\{Q_{\tau,t}(Y|A,L)\}+ (A-E(A|L))\frac{d}{dt}\left(g\{Q_{\tau,t}(Y|A,L)\}\right)\Big\lvert_{t=0}\\ 
    &=-\frac{\mathbbm{1}_{\Tilde{l}}(L)}{f(L)}(\Tilde{a}-E(A|L))g\{Q_{\tau,t}(Y|A,L)\} + \{A-E(A|L)\}g'\{Q_\tau(Y|A,L)\}\frac{d}{dt}Q_{\tau,t}(Y|A,L)\Big\lvert_{t=0}\\
    &=-\frac{\mathbbm{1}_{\Tilde{l}}(L)}{f(L)}(\Tilde{a}-E(A|L))g\{Q_{\tau,t}(Y|A,L)\} + \{A-E(A|L)\}g'\{Q_\tau(Y|A,L)\}\mathbbm{1}_{\Tilde{a},\Tilde{l}}(A,L)\frac{\tau-I\{\Tilde{y}\leq Q_\tau(Y|A,L)\}}{f(Q_\tau(Y|A,L),A,L)}
\end{align*}
such that the efficient influence function is
\begin{align*}
    & \phi(Y,A,L)=\frac{A-E(A|L)}{E[\{A-E(A|L)\}^2]}\bigg[g\{Q_\tau(Y|A,L)\}- E[(g\{Q_\tau(Y|A,L)\}|L]\\
    &\qquad\qquad\qquad+ \frac{\tau-I\{Y\leq Q_\tau(Y|A,L)\}}{f\{Q_\tau(Y|A,L)|A,L\}}g'\{Q_\tau(Y|A,L)\}- \Theta_\tau \{A-E(A|L)\} \bigg].
\end{align*}
The proposed estimators can be easily adjusted for the estimation of $\Theta_\tau$.
\end{document}